\documentclass[journal]{IEEEtran}

\usepackage{cite}
\usepackage{easybmat}
\usepackage{stfloats}
\usepackage{amsmath,amssymb,amsfonts}
\usepackage{graphicx}
\usepackage{textcomp}
\usepackage[linesnumbered,lined,boxed,commentsnumbered, ruled]{algorithm2e}
\usepackage{epstopdf}
\def\BibTeX{{\rm B\kern-.05em{\sc i\kern-.025em b}\kern-.08em
    T\kern-.1667em\lower.7ex\hbox{E}\kern-.125emX}}

\usepackage{longtable}
\usepackage{caption}
\usepackage{subcaption}

\usepackage{hhline}
\usepackage[dvipsnames]{xcolor}

\newtheorem{thm}{Theorem}

\newtheorem{lem}{Lemma}
\newtheorem{defi}{Definition}
\newtheorem{ex}{Example}

\newenvironment{proof}{ \paragraph*{\hspace{-1em}Proof}}{\hfill$\square$}

\renewcommand{\vec}[1]{\ensuremath{\boldsymbol{#1}}}
\newcommand{\be}{\begin{equation}}
\newcommand{\ee}{\end{equation}}
\newcommand{\ba}{\begin{array}}
\newcommand{\ea}{\end{array}}
\newcommand{\bea}{\begin{eqnarray}}
\newcommand{\eea}{\end{eqnarray}}
\newcommand{\A}{\vec{A}}    
\newcommand{\As}{\vec{A}_{\mathcal{S}}}    
\newcommand{\Ws}{\vec{W}_{\mathcal{S}}}    

\newcommand{\X}{\vec{X}}    
\newcommand{\W}{\vec{W}}    
\newcommand{\WAX}{\vec{WAX}}    
\newcommand{\AX}{\vec{AX}}    
\newcommand{\ra}[1]{\mathrm{rank}\left( #1\right)}    
\renewcommand{\H}{\vec{H}}

\begin{document}
%
\title{Trade-offs in Decentralized Multi-Antenna Architectures: The WAX Decomposition}
%
%
%

\author{Juan Vidal Alegr\'ia,
        Fredrik Rusek, and Ove Edfors
        \thanks{This paper is built upon previous results presented at the 2020 IEEE ICC conference \cite{icc2020}.}
}

\maketitle

\begin{abstract}
Current research on multi-antenna architectures is trending towards increasing the amount of antennas in the base stations (BSs) so as to increase the spectral efficiency. As a result, the interconnection bandwidth and computational complexity required to process the data using centralized architectures is becoming prohibitively high. Decentralized architectures can reduce these requirements by pre-processing the data before it arrives at a central processing unit (CPU). However, performing decentralized processing introduces also cost in complexity/interconnection bandwidth at the antenna end which is in general being ignored. This paper aims at studying the interplay between level of decentralization and the associated complexity/interconnection bandwidth requirement at the antenna end. To do so, we propose a general framework for centralized/decentralized architectures that can explore said interplay by adjusting some system parameters, namely the number of connections to the CPU (level of decentralization), and the number of multiplications/outputs per antenna (complexity/interconnection bandwidth). We define a novel matrix decomposition, the WAX decomposition, that allows information-lossless processing within our proposed framework, and we use it to obtain the operational limits of the interplay under study. We also look into some of the limitations of the WAX decomposition.

\end{abstract}


\section{Introduction}
\label{section:intro}
\IEEEPARstart{M}{ulti-antenna} architectures have been widely employed since they were first introduced in the 1990s and they still remain a popular research topic. The main reason is the enormous improvements in data rate and reliability coming from exploiting space-division multiplexing and diversity. Current research and development on multi-antenna architectures is trending towards scaling up the number of antennas so as to increase the spatial resolution, thus increasing the spectral efficiency by serving several users in the same time-frequency resource. Furthermore, the exploitation of millimeter-wave spectrum in modern communications \cite{5G} also justifies the increase in the number of antennas. The reason is that the huge path-loss associated to these frequencies when the electrical size of the antennas is kept constant needs to be overcome by focusing the transmitted energy more effectively \cite{mmW}. Massive multiple-input multiple-output (MIMO) \cite{marzetta,rusek} and large intelligent surface (LIS) \cite{husha_data} are some examples of the trend towards increasing the number of antennas, where massive MIMO considers base stations (BSs) with hundreds of antennas while LIS extends this concept even further by considering whole walls of electromagnetically active material.

Massive MIMO is already a reality and several prototypes have been developed and tested, such as \cite{lumami,argos,bigstation}. In the prototypes presented in \cite{lumami,bigstation}, the use of centralized processing leads to huge data-rates between the antennas and the central processing unit (CPU), which limits the scalability of the system as the number of antennas grows. This fact is also noticed in \cite{argos}, which sacrifices performance by relying on simple decentralized beam-forming to favor scalability. The scalability issue is likely to be exacerbated if we consider LIS, where we can think of practical deployments consisting of walls equipped with an even larger number of antennas than massive MIMO (the continuous surfaces are approximately equivalent to the discrete surfaces when the sampling is dense enough, as observed in \cite{husha_data,lis_hw_imp}). Other technologies that are gaining popularity and are likely to face scalability issues include cell-free massive MIMO \cite{cell_free,sc_cell_free,zhang_prosp}, or intelligent reflecting surfaces (IRS) \cite{wu_irs,chen_irs,zhang_prosp}. We will base our study in a general multi-antenna architecture so that it can be easily extended to more specific applications, such as the ones previously mentioned.

There is a current trend towards more decentralized architectures \cite{cavallaro,larsson,li_tradeoffs,isit_2019,muris,jesus,vtc, zhang, amiri} so as to reduce the information transmitted to the CPU. The idea is to carry out pre-processing of the data at the antenna end (or close to it), so that the CPU does not need to have access to all the information required to decode raw data. Available literature on decentralized massive MIMO proposes a wide range of solutions from fully-decentralized architectures \cite{jesus,muris,vtc, cavallaro}, where channel state information (CSI) does not have to be available at the CPU, to partially decentralized architectures, where some of the processing tasks are distributed, but either full \cite{larsson,amiri} or partial CSI \cite{isit_2019} is available at the CPU. 

In this paper we do not address the problem of decentralized CSI estimation; we assume that perfect CSI is available at the CPU. We rely on the fact that CSI estimation does not limit the overall level of decentralization within our framework since it needs to be carried out only once per coherence block. Thus, CSI estimation takes a minor fraction of the coherence time, and the estimated CSI can be then used for the data phase throughout the rest of the coherence block without affecting the level of decentralization. However, the problem of estimating and sharing CSI in an efficient and scalable way within our framework remains as future work.

In \cite{isit_2019} it is argued that an architecture is decentralized enough if it does not need extra hardware apart from the minimum required during the payload data phase. It also states that the volume of data transferred during the data phase has to be independent of the number of antennas at the base station. However, as happens in \cite{isit_2019,muris,jesus,vtc}, in order to reduce this volume of data (related to the number of connections to the CPU) and make it independent of the number of antennas, each antenna has to provide a number of outputs that scales with the number of users. Note that in a centralized architecture we would have only one output per antenna corresponding to one input to the CPU per antenna. We notice the existence of an interesting trade-off between the number of connections to the CPU and the number of outputs from each antenna.

The main goal of this paper is to study the interplay between level of decentralization and the corresponding increase in decentralized processing complexity for multi-antenna architectures. We measure the level of decentralization as the number of connections to a CPU required during the data phase, and we measure the decentralized processing complexity as the number of multiplications/outputs per antenna (or antenna panel)\footnote{Antenna panel refers to a group of co-located antennas.} required to achieve a given level of decentralization. So as to study this trade-off, we present a general framework for a multi-antenna architecture which allows us to change the level of decentralization and complexity by adjusting some system parameters. In this framework, the antennas are grouped into panels of a given number of antennas (this number can be 1 as in \cite{icc2020}). Distributed processing is applied by applying a linear transformation to the inputs of each antenna panel, which generates a given number of outputs (complexity). The outputs are then combined using a combining module that is connected to the CPU using a fixed number of CPU inputs (level of decentralization).

To the best of our knowledge, the presented trade-off has not been explored in the available literature and the results we present are completely new. In \cite{li_tradeoffs} trade-offs between different decentralized architectures, algorithms, and data precision levels are studied. However, these trade-offs are mainly systematic while we are interested in fundamental trade-offs where information rates are maintained with respect to typical centralized systems. Thus, our framework focuses on complex baseband processing, and we assume that the detection is always performed in a CPU.

This paper extends the work presented in \cite{icc2020}. The list of contributions are summarized next:
\begin{itemize}
    \item We present a novel general framework for multi-antenna systems that allows us to consider different levels of decentralization and complexity by adjusting the system parameters. This general framework accounts for typical centralized architectures, decentralized architectures such as the one presented in \cite{isit_2019}, and hybrid architectures in between those two. In \cite{icc2020}, a similar framework to the one studied in this paper is considered. However, the architecture presented in the current work is more general since it adds the possibility of arranging the antennas into panels.
    \item We define the WAX decomposition, a novel matrix decomposition that allows us to define and exploit the trade-offs within our general framework while achieving information-lossless processing. In \cite{icc2020} the WAX decomposition is already introduced, but in the current work we adapt it to use it in a more general framework. Furthermore, we present novel results on the validity and application of said decomposition, e.g., Theorem~\ref{th:A_randH}.
    \item We present the trade-off in terms of number of multiplications/outputs per antenna panel and number of connections to the CPU for achieving information-lossless processing within our general framework. This trade-off is first studied in \cite{icc2020}, but considering only the less general version of the current framework.
    \item We study through simulations the cost of obtaining simple combining modules that accept WAX decomposition within our general framework.
    \item We study complexity limitations for the combining modules within our general framework that accept WAX decomposition.
    \item We study the information-loss associated to operating within our framework when WAX decomposition is not available.
    \item We present a simple non-optimal solution for determining the distributed processing to be applied whenever WAX decomposition is not available. A deeper research on more effective solutions remains as future work.
\end{itemize}

The rest of the paper is organized as follows. Section~\ref{section:model} presents the general framework under study, as well as the system model and problem formulation. In Section~\ref{section:wax} we present the WAX decomposition, which allows the application of information-lossless processing within our general framework. In Section~\ref{section:sparse_A} we study the problem of defining a simple, but valid, combining network for our general framework. Section~\ref{section:tradeoff} presents some discussion on the resulting trade-offs, and some examples of the usage of WAX decomposition. Section~\ref{section:lossy} explores broadly the case where WAX decomposition is not available and other information-lossy processing has to be applied to maximize the data-rates within our framework. We conclude the paper in Section~\ref{section:conclusions} with a summary of the contributions and future work.

Notation: In this paper, lowercase, bold lowercase and bold uppercase
letters stand for scalars, column vectors and matrices, respectively. When using the mutual information operator, $I(.;.)$, bold uppercase in the sub-scripts refers to random vectors instead of their realizations. The operations $(\cdot)^T$, $(\cdot)^*$ and $(\cdot)^H$ denote transpose, conjugate, and conjugate transpose, respectively. The operation $\mathrm{diag}(\cdot)$ outputs a block diagonal matrix with the input matrices/vectors as the diagonal blocks. The operator $\mathrm{vec}(\cdot)$ transforms a matrix into a vector by concatenating its columns. $\mathbf{I}_i$ corresponds to the identity matrix of size $i$, $\boldsymbol{1}_{i\times j}$ denotes the $i \times j$ all-ones matrix, and $\boldsymbol{0}_{i \times j}$ denotes the $i \times j$ all-zeros matrix. In this paper, a randomly chosen matrix corresponds to a realization of a random matrix whose elements are driven from a continuous probability distribution function.

\section{System Model} 
\label{section:model}
Let us consider $K$ single-antenna users transmitting to an $M$-antenna BS through a narrow-band channel. The $M\times 1$ received complex vector, $\boldsymbol{y}$, can be expressed as
\begin{equation}\label{eq:ul_model}
\boldsymbol{y} = \boldsymbol{H}\boldsymbol{s} + \boldsymbol{n},
\end{equation}
where $\boldsymbol{H}$ is the channel matrix of dimension $M\times K$, $\boldsymbol{s}$ is the $K \times 1$ vector of symbols transmitted by the users, and $\boldsymbol{n}$ is a zero-mean complex white Gaussian noise vector with sample variance $N_0$. The $M$ antennas are divided into panels of $N$ antennas; $M/N=P$ is thus restricted to integer values. Each panel,  $p\in\{1,\dots, P\}$, multiplies the received vector, $\boldsymbol{y}_p = [y_{(p-1)N+1} \, \dots \, y_{p N}]$, by an $L\times N$ matrix, $\boldsymbol{W}_p^H$. The aggregated outputs are combined through a fixed $T\times L P$ matrix, $\boldsymbol{A}^H$. The resulting vector is forwarded to a CPU, which can apply further processing. For our analysis, we will assume that the CPU can multiply the incoming matrix by a $K\times T$ matrix $\boldsymbol{X}^H$ to be able to express the equivalent matrix in the form of already known linear equalizers. However, our main interest is to maximize the information rate at which the users transmit to the CPU, so the last step is not required for the analysis since it will not increase this information rate (recall the data-processing inequality \cite{inf_th}). Also, we define the previous matrices using conjugate transpose so as to ease upcoming notation. Fig.~\ref{fig:gen_arch} shows a block diagram of the general framework under study.

\begin{figure}[h]
	\centering
	\includegraphics[scale=0.53]{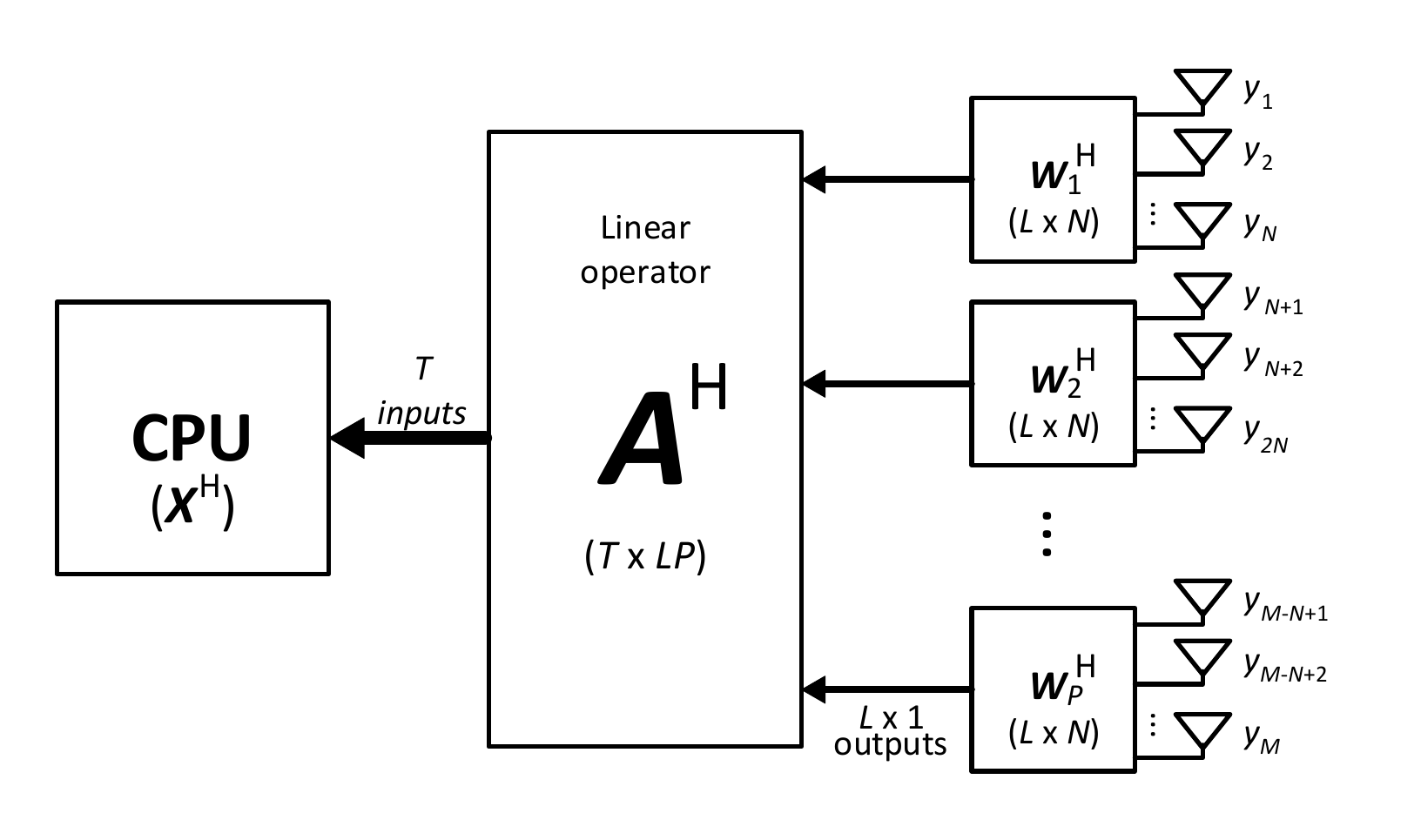}
\caption{General framework considered in this paper.}
	\label{fig:gen_arch}
\end{figure}

The post-processed vector can be expressed as
\begin{equation}\label{eq:z_proc}
\boldsymbol{z} = \boldsymbol{X}^H \boldsymbol{A}^H \boldsymbol{W}^H \boldsymbol{y},
\end{equation}
where $\boldsymbol{W}$ is an $M \times L P$ block diagonal matrix of the form
\begin{equation}\label{eq:w}
    \boldsymbol{W} = \mathrm{diag}\left(\boldsymbol{W}_1,\boldsymbol{W}_2,\dots,\boldsymbol{W}_P\right).
\end{equation}
We assume that the matrices $\boldsymbol{W}_p$, and $\boldsymbol{X}$ can be tuned for every channel realization, while the matrix $\boldsymbol{A}$ is fixed. We can interpret $\boldsymbol{A}$ as the matrix associated to a hardware combining network that can be predesigned, but it cannot be changed once deployed. Table \ref{table:param} shows a classification of the different system parameters considered within our framework.

For tractability we assume that the channel matrix, $\boldsymbol{H}$, is perfectly available at the BS. Thus, we can obtain $\boldsymbol{W}_p$, and $\boldsymbol{X}$ as a function of said matrix. However, the consequences of having an error in the estimation of $\boldsymbol{H}$ due to imperfect CSI would not be enhanced by the framework under study. In fact, the information-loss associated to having imperfect CSI within our framework would not differ from the case of having imperfect CSI in a typical centralized architecture. The reason is that we could still apply an approximation of the spatially-matched filter (MF) within our framework using the imperfectly estimated channel matrix, as we will be able to understand from the upcoming analysis.

\begin{table}
\centering
\caption{System parameters}
\label{table:param}
\scalebox{1.3}{
\begin{tabular}{ |c|c| } 
\hline
 Given parameters & Trade-off parameters \\
 $M$, $K$ & $L$, $T$, $N$, $P$ \\
\hline
Tunable parameters & Parameters fixed by design \\
$\boldsymbol{W}$, $\boldsymbol{X}$ & $\boldsymbol{A}$\\
\hline
\end{tabular}}
\end{table}

For a fixed $\boldsymbol{A}$, we are interested in maximizing the information rate at which the users can transmit, i.e., we would like to solve the maximization problem
\begin{equation}\label{eq:maximize}
\begin{aligned}
& \underset{\boldsymbol{X},\{\boldsymbol{W}_p\}_{p=1}^P}{\text{maximize}}
& & I_{\boldsymbol{Z},\boldsymbol{S}}(\boldsymbol{z}; \boldsymbol{s}).
\end{aligned}
\end{equation}
More specifically, we will explore the cases where the maximization results in $I_{\boldsymbol{Z},\boldsymbol{S}}(\boldsymbol{z},\boldsymbol{s})=I_{\boldsymbol{Y},\boldsymbol{S}}(\boldsymbol{y},\boldsymbol{s})$. From the data-processing inequality \cite{inf_th} we have
\begin{equation}
\begin{aligned}
    I_{\boldsymbol{Z},\boldsymbol{S}}(\boldsymbol{z};\boldsymbol{s})&\leq I_{\boldsymbol{Y},\boldsymbol{S}}(\boldsymbol{A}^H\boldsymbol{W}^H\boldsymbol{y};\boldsymbol{s}) \\
    &\leq I_{\boldsymbol{Y},\boldsymbol{S}}(\boldsymbol{y};\boldsymbol{s}).
\end{aligned}
\end{equation}
This means that the application of $\boldsymbol{X}^H$ at the CPU cannot possibly increase the information rate and, as we mentioned before, it is just a manipulation to adapt the dimensions. Furthermore, assuming $\boldsymbol{s}\sim \mathcal{CN}(\boldsymbol{0}_{K\times 1},P_s \mathbf{I}_K)$ so that the mutual information is maximized (and thus coincides with the capacity), we have \cite{mimo}
\begin{equation}
    I_{\boldsymbol{Y},\boldsymbol{S}}(\boldsymbol{y};\boldsymbol{s})=\log \det \left( \mathbf{I}_M+\frac{P_s}{N_0}\boldsymbol{H}\boldsymbol{H}^H \right).
\end{equation}
Therefore, we can state that, if we are able to achieve
\begin{equation}
I_{\boldsymbol{Z},\boldsymbol{S}}(\boldsymbol{z};\boldsymbol{s})=I_{\boldsymbol{Y},\boldsymbol{S}}(\boldsymbol{y};\boldsymbol{s}),
\end{equation}
the information rate is maximized and the equivalent processing is information-lossless.

One of the main scopes of this paper is to study the conditions, in terms of constraints on the system parameters, for our framework to be able to perform information-lossless processing. The following lemma will be helpful.

\begin{lem}
\label{lem:cap_WAX}
Considering the presented framework, the equality
\begin{equation}
I_{\boldsymbol{Z},\boldsymbol{S}}(\boldsymbol{z};\boldsymbol{s})=\log \det \left( \mathbf{I}_M+\frac{P_s}{N_0}\boldsymbol{H}\boldsymbol{H}^H \right),
\end{equation}
is fulfilled if and only if we can find a $\boldsymbol{W}$ and $\boldsymbol{X}$ such that
\begin{equation}
\boldsymbol{WAX}= \boldsymbol{H}.
\end{equation}
\end{lem}
\begin{proof}
Let $R(\boldsymbol{y})$ be any sufficient statistic for $\boldsymbol{s}$ (which means that it is information-lossless \cite{inf_th}). Then, for any $S\geq K$ there exists a full-rank $S\times K$ matrix $\widetilde{\boldsymbol{X}}$ such that
\begin{equation}
    R(\boldsymbol{y}) = \widetilde{\boldsymbol{X}}\boldsymbol{H}^H \boldsymbol{y}.
\end{equation}
Therefore, for $S=K$, we have that $\boldsymbol{X}^H \boldsymbol{A}^H \boldsymbol{W}^H \boldsymbol{y}$ is a sufficient statistic if and only if
\begin{equation}
    \boldsymbol{X}^H \boldsymbol{A}^H \boldsymbol{W}^H \boldsymbol{y} = \widetilde{\boldsymbol{X}}\boldsymbol{H}^H \boldsymbol{y},
\end{equation}
which leads to
\begin{equation}
    \widetilde{\boldsymbol{X}}^{-1}\boldsymbol{X}^H \boldsymbol{A}^H \boldsymbol{W}^H  = \boldsymbol{H}^H,
\end{equation}
but since $\widetilde{\boldsymbol{X}}^{-1}$ can be absorbed by $\boldsymbol{X}$, we have
\begin{equation}\label{eq:mf_xaw}
    \boldsymbol{X}^H \boldsymbol{A}^H \boldsymbol{W}^H  = \boldsymbol{H}^H,
\end{equation}
which proves the lemma by taking the conjugate transpose.
\end{proof}

Lemma~\ref{lem:cap_WAX} gives an important result to understand when the maximization in \eqref{eq:maximize} achieves information-lossless processing. This lemma will be the basis in Section~\ref{section:wax} for solving the maximization problem \eqref{eq:maximize} whenever information-lossless processing is available within our framework. Note that \eqref{eq:mf_xaw} would correspond to applying MF within our framework, which is well known to be an information-lossless transformation if optimum processing is applied thereafter.

\subsection{Notes on the downlink scenario}
Throughout this paper we focus on the uplink scenario for improved clarity. However, this work can straightforwardly be extended to a downlink scenario. We next provide a few details.

Let us assume channel reciprocity, which eases notation and remarks the equivalence with the uplink scenario.\footnote{In case of non-reciprocity, the presented framework would still be valid as long as the base station has access to the downlink channel matrix.} This means that, given the uplink equation \eqref{eq:ul_model}, the corresponding downlink equation for the vector received by the users is as follows
\begin{equation}\label{eq:dl_model}
\boldsymbol{y}_\mathrm{d} = \boldsymbol{H}^T\boldsymbol{z}_\mathrm{d} + \boldsymbol{n}_\mathrm{d},
\end{equation}
where $\boldsymbol{y}_\mathrm{d}$ is now a $K \times 1$ vector with entries associated to the complex baseband signal seen by each user, $\boldsymbol{z}_\mathrm{d}$ is the $M\times 1$ precoded vector sent by the antennas, and $\boldsymbol{n}_\mathrm{d}$ is the corresponding noise vector. We assume that the same framework as in Fig.~\ref{fig:gen_arch} applies for the downlink, but changing the arrows from left to right. This implies that the linear operators (including the combining module) are assumed to be able to use inputs as outputs, and vice versa. The precoded vector would then be defined as
\begin{equation}
\boldsymbol{z}_\mathrm{d} = \W^*\A^*\X^*\boldsymbol{s}_\mathrm{d},
\end{equation}
where $\boldsymbol{s}_\mathrm{d}$ is now a $K\times 1$ vector with complex entries associated to the signals intended for each user. In this case, Lemma~\ref{lem:cap_WAX} does not apply unless we assume that the users can collaborate with each other. However, we can achieve lossless precoding with respect to standard centralized linear precoding schemes, since we can still apply typical precoding schemes within our framework such as MF, zero-forcing (ZF) or minimum mean squared error (MMSE). Furthermore, with the assumption of channel reciprocity in place, the $\W$ matrices could be kept fixed from the uplink processing ($\A$ is still fixed by design) and $\X$ can be adapted to apply the desired precoding scheme. In the case of non-reciprocity of the channel, both $\X$ and $\W$ have to be recomputed for applying the desired precoding based on the downlink channel.

\subsection{Previously studied architectures within our framework}
There are several multi-antenna architectures that could be represented within our framework, which further motivates our study. In this case, what we mean with "represent an architecture within our framework" is  that there is a combination of design variables within our framework that gives the same processing. The most obvious architecture that fulfills this is a typical centralized $M$-antenna architecture, e.g., centralized massive MIMO systems. In this case all the antennas are directly connected to a CPU, which corresponds to having one antenna per panel, and all antennas connected directly to the CPU, i.e., $N=1$, $L=1$, $\boldsymbol{W}_p=1$ (scalar), $T=M$, $\boldsymbol{A}=\mathbf{I}_M$. This architecture is depicted in Fig.~\ref{fig:arch_old} (left).

Another architecture that can be represented within our framework is the decentralized massive MIMO architecture from \cite{isit_2019}. In this case, there is also one antenna per panel, but antenna $m$ multiplies its input by the $K \times 1$ local channel vector, $\boldsymbol{h}_m$, and the result is summed over all the antennas so that the size of the vector transmitted to the CPU coincides with the number of users, $K$. Thus, this architecture can be represented within our framework by setting $N=1$, $L=K$, $\boldsymbol{W}_p=\boldsymbol{h}_p$ (vector), $T=K$, $\boldsymbol{A}=\left[ \mathbf{I}_K \;\mathbf{I}_K \,\dots \,\mathbf{I}_K \right]^T$, which corresponds to MF if no $\boldsymbol{X}$ is applied at the CPU. This architecture is depicted in Fig.~\ref{fig:arch_old} (right). In \cite{isit_2019} it is claimed that, for a system to be decentralized, the volume of data transmitted to the CPU during the data phase should not scale with $M$. However, the proposed solution reduces this scaling to $K$ by increasing the number of multiplications/outputs per antenna to $K$, which increases the decentralized processing complexity. Our framework allows us to freely adjust these parameters.

\begin{figure*}[h]
	\centering
	\includegraphics[scale=0.53]{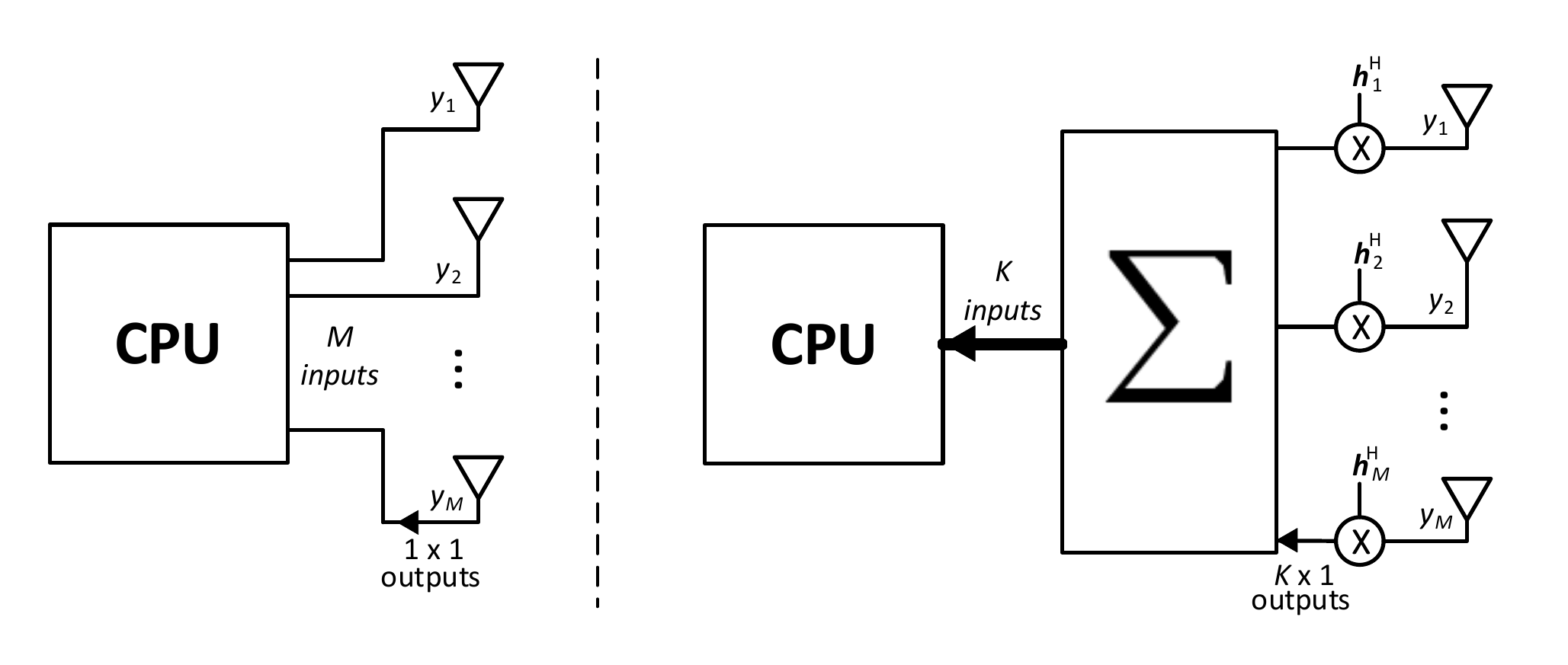}
\caption{Architecture of centralized massive MIMO (left), and decentralized massive MIMO from \cite{isit_2019} (right).}
	\label{fig:arch_old}
\end{figure*}

Comparing the two architectures from Fig.~\ref{fig:arch_old}, where in both cases information-lossless processing can be applied, we immediately identify a trade-off between the number of connections to the CPU and the number of outputs per antenna panel as depicted in Fig.~\ref{fig:outputs} (left). Note that, in this case, we have a single antenna per panel, i.e., $N=1$. In our framework, these two quantities, which are traded-off against each other, correspond to $T$ and $L$, respectively. We can also see the trade-off in terms of multiplications per antenna as in Fig.~\ref{fig:outputs} (right), which can lead to a fairer comparison if we consider panels with more than 1 antenna. The main reason is that, for $N>1$, the number of outputs can be maintained \cite{feedforward}, but the complexity still increases due to a higher number of multiplications.

If we look into other decentralized architectures, such as the ones presented in \cite{feedforward, jesus, muris, li_tradeoffs, vtc}, we can also represent them within our framework and identify the trade-off between inputs to the CPU and multiplications/outputs per panel (in these cases the trade-off may not lead to information-lossless processing). For example, the architectures from \cite{jesus, muris, vtc} would all lead to the decentralized point in the plots from Fig.~\ref{fig:outputs} since they all apply $K\times 1$ filters in each antenna during the data phase and generate $K$ outputs per panel during the data phase ($N=1$ in these architectures). The same is true in \cite{feedforward} if we consider outputs per panel, although if we look at multiplications per panel we should scale the number of outputs by the number of antennas per panel, $N> 1$. Note that, even though \cite{feedforward} presents two different architectures with different level of decentralization, within our framework they lead to the same value of the trade-off; the only difference is that in the fully-decentralized architecture $\boldsymbol{X}$ would not be applied. In \cite{li_tradeoffs}, however, we would not be able to represent the fully-decentralized architecture within our framework since detection is performed before reaching the CPU, while our framework only considers complex baseband processing.

\begin{figure}
     \centering
     \begin{subfigure}[b]{0.23\textwidth}
         \centering
         \includegraphics[scale=0.34]{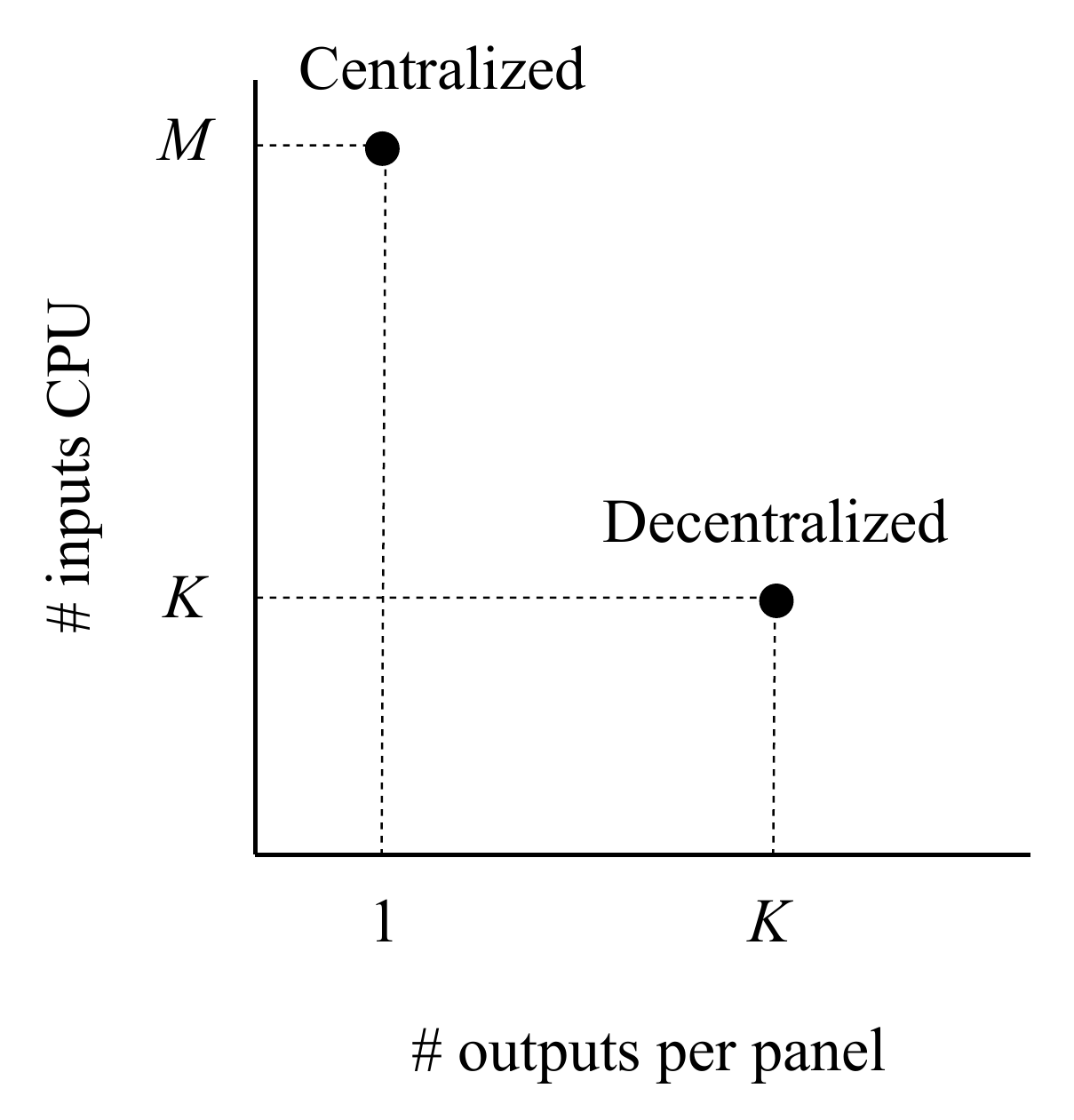}
     \end{subfigure}
     \hfill
     \begin{subfigure}[b]{0.23\textwidth}
         \centering
         \includegraphics[scale=0.34]{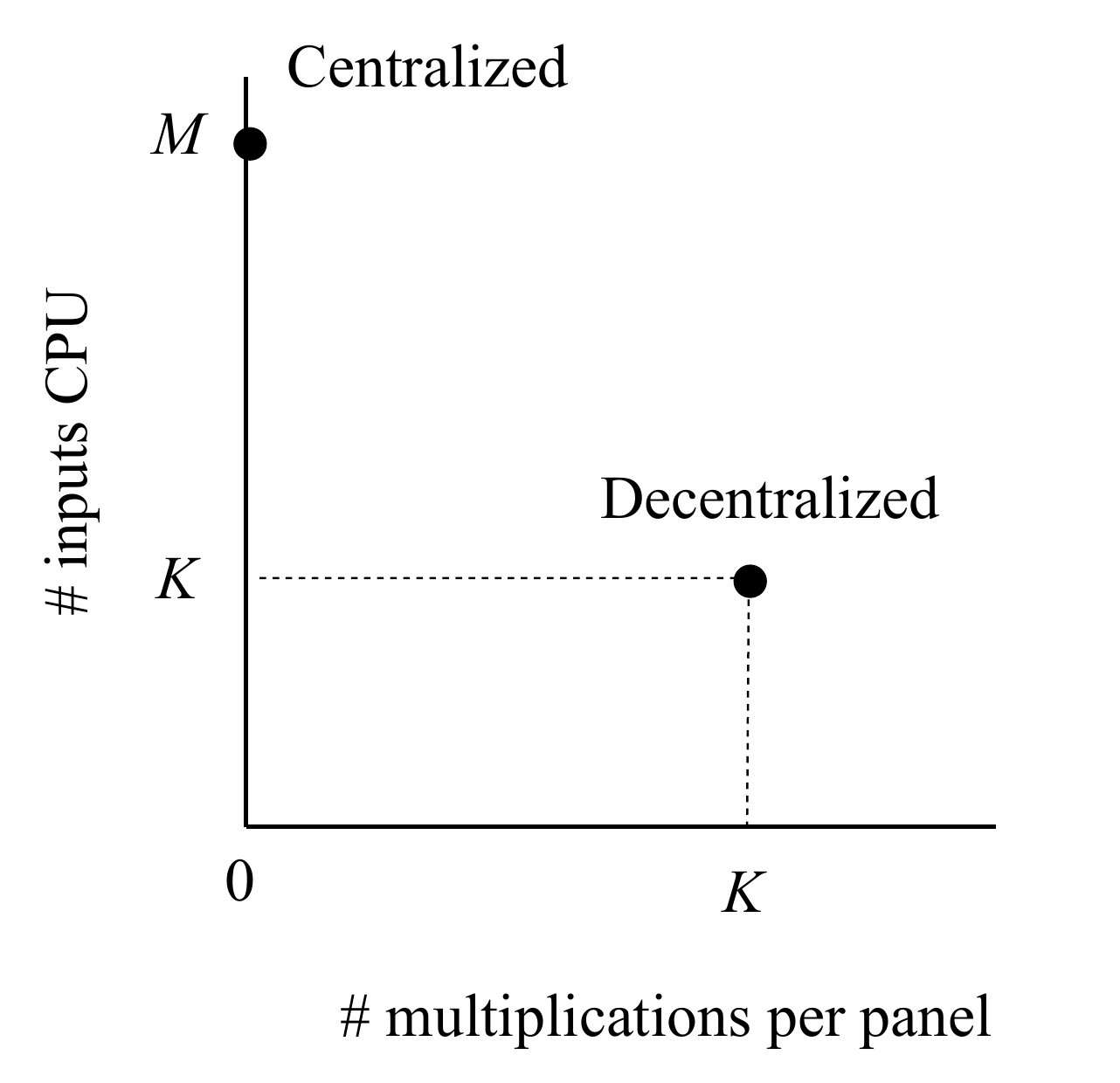}
     \end{subfigure}
        \caption{Number of inputs to the CPU (related to $T$) v.s. number of outputs (left)/number of multiplications (right) per antenna (related to $L$) for the architectures depicted in Fig.~\ref{fig:arch_old}. Available literature does not address the behavior in between 1 and $K$ outputs per antenna.}
        \label{fig:outputs}
\end{figure}

\section{WAX Decomposition}
\label{section:wax}
Given an arbitrary $M \times K$ matrix, $\boldsymbol{H}$, and a fixed $L P \times T$ matrix, $\boldsymbol{A}$, we define the WAX decomposition of $\boldsymbol{H}$, whenever it exists, as
\begin{equation} \label{WAX}
  \boldsymbol{H} =  \boldsymbol{W}\boldsymbol{A}\boldsymbol{X}.
\end{equation}
This decomposition relates directly to Lemma~\ref{lem:cap_WAX}, i.e., $\boldsymbol{W}$ has the structure defined in (\ref{eq:w}) and $\boldsymbol{X}$ is a $T \times K$ matrix. Furthermore, the existence of this matrix decomposition for any channel realization would assure that our framework can apply information-lossless processing.

The next lemma will allow us to restrict our study to the case where $N=L$ since this case can span a major part of the general problem through some manipulation.
\begin{lem}
\label{lem:N_L}
Assume $M/L$ and $L/N$ are integer valued. Consider some fixed $T \times K$ matrix $\boldsymbol{X}$, an $M \times K$ matrix $\boldsymbol{H}$, an $M \times T$ matrix $\widetilde{\boldsymbol{A}}$, and an $L P \times T$ matrix $\boldsymbol{A} = \boldsymbol{T}_{L,N}\widetilde{\boldsymbol{A}}$, with
\begin{equation}\label{eq:WAX_N}
    \boldsymbol{T}_{L,N}=\mathbf{I}_{\frac{M}{L}}\otimes \left( \boldsymbol{1}_{\frac{L}{N} \times 1} \otimes \mathbf{I}_L \right).
\end{equation}
Then, the statement
\begin{equation}
    \exists \boldsymbol{W}\, : \, \boldsymbol{H} =  \boldsymbol{W}\boldsymbol{A}\boldsymbol{X},
\end{equation}
where $\boldsymbol{W}$ is defined as in \eqref{eq:w}, is true if and only if
\begin{equation}\label{eq:WAX_L}
    \exists \widetilde{\boldsymbol{W}}\, : \, \boldsymbol{H} =  \widetilde{\boldsymbol{W}}\widetilde{\boldsymbol{A}}\boldsymbol{X},
\end{equation}
where $\widetilde{\boldsymbol{W}}=\mathrm{diag}(\widetilde{\boldsymbol{W}}_1,\dots, \widetilde{\boldsymbol{W}}_{M/L})$, and $\widetilde{\boldsymbol{W}}_p$ are $L \times L$ matrices $\forall p\in \{1,\dots, M/L \}$, i.e., $\widetilde{\boldsymbol{W}}$ has the same structure as \eqref{eq:w} for $N=L$.
\end{lem}
\begin{proof}
The right implication can be immediately seen by choosing $\widetilde{\boldsymbol{W}}=\boldsymbol{W}\boldsymbol{T}_{L,N}$ and noticing that this matrix already has the required structure.

For the left implication it is enough to check that, for any $\widetilde{\boldsymbol{W}}$, there exists some $\boldsymbol{W}$ that fulfills $\boldsymbol{W}\boldsymbol{T}_{L,N} = \widetilde{\boldsymbol{W}}$, both $\boldsymbol{W}$ and $\widetilde{\boldsymbol{W}}$ with their corresponding structure. In this case, there is no linear transformation that gets $\boldsymbol{W}$ from $\widetilde{\boldsymbol{W}}$, but we can construct it as follows. Let us take $\widehat{\boldsymbol{W}}= \widetilde{\boldsymbol{W}}\boldsymbol{T}_{L,M}$, with $\boldsymbol{T}_{L,M}=[\mathbf{I}_L,\dots, \mathbf{I}_L]^T$. We can then select the diagonal blocks of $\boldsymbol{W}$ to be $\boldsymbol{W}_p = [\widehat{\boldsymbol{w}}_{(p-1)N+1}^T, \dots, \widehat{\boldsymbol{w}}_{(p-1)N}^T]^T$, where $\widehat{\boldsymbol{w}}_m$, $m=\{1,\dots,M\}$, are the rows of $\widehat{\boldsymbol{W}}$. Thus, there is a one-to-one mapping between all possible $\widetilde{\boldsymbol{W}}$ and $\boldsymbol{W}$ that fulfills $\boldsymbol{W}\boldsymbol{T}_{L,N} = \widetilde{\boldsymbol{W}}$.
\end{proof}

Lemma~\ref{lem:N_L} also implies that, whenever the assumptions apply, the trade-off between the number of outputs per panel and the number of inputs to the CPU required ($L$ and $T$, respectively) does not depend on the number of antennas per panel ($N$). However, the number of multiplications per panel would scale with $N$. An important constraint of Lemma~\ref{lem:N_L} is that we need to have $L\geq N$, which can be easily checked to be a requirement for having information-lossless processing within our framework unless $L=K$ (leading to the trivial solution of setting $\boldsymbol{W}_p$ as the local MFs).

The rest of this paper assumes $N=L$, unless otherwise stated, due to the intrinsic generality of this case. This means that $\boldsymbol{W}$ in \eqref{WAX} is a square $M \times M$ block diagonal matrix given by \eqref{eq:w}, with $P=M/L$, and containing square $L\times L$ blocks. The physical implication of this restriction is that all the panels in Fig.~\ref{fig:gen_arch} have the same number of antennas and outputs. However, as seen in the proof of Lemma~\ref{lem:N_L}, we can construct almost any other case from the square case by applying some transformations, with the only limitation that $M/L$ and $L/N$ must evaluate to integer values, as well as that $\boldsymbol{A}$ must be constructed using a specific structure. For simplicity, we will use the same notation as in the general case \eqref{WAX}, but considering the new restriction on the dimensions.

\subsection{Solving WAX}
Let us assume for now that there exists the WAX decomposition of $\boldsymbol{H}$, we later investigate when this is the case. We provide next a step-by-step solution for practical computation of \eqref{WAX}, i.e., for obtaining the matrices $\boldsymbol{W}$ and $\boldsymbol{X}$ in \eqref{WAX} based on the current channel realization, $\boldsymbol{H}$, and the fixed combining network, $\boldsymbol{A}$. This step-by-step solution will also be useful to set the ground for our main result on the applicability of WAX decomposition.

Let us express $\boldsymbol{H}=[\boldsymbol{H}_1^T \, \boldsymbol{H}_2^T \, \dots \boldsymbol{H}_P^T]^T$ and $\boldsymbol{A}=[\boldsymbol{A}_1^T \, \boldsymbol{A}_2^T \, \dots \boldsymbol{A}_P^T]^T$, where $\boldsymbol{H}_p$ and $\boldsymbol{A}_p$ $\forall p \in \{1,\dots,P\}$ are $L \times K$ and $L \times T$ blocks, respectively. The following lemma will be useful.
\begin{lem} \label{lemeqv}
For all matrices $\boldsymbol{H}$ satisfying $\mathrm{rank}(\boldsymbol{H}_n)=L$ there exists a block diagonal matrix $\boldsymbol{W}$ and a matrix $\boldsymbol{X}$ such that $\boldsymbol{W}\boldsymbol{A}\boldsymbol{X}=\boldsymbol{H}$ if and only if there exists  a block diagonal invertible matrix $\hat{\boldsymbol{W}}$ with the same form as $\boldsymbol{W}$ that fulfills $\boldsymbol{A}\boldsymbol{X}-\hat{\boldsymbol{W}}\boldsymbol{H}=\boldsymbol{0}_{M\times K}$.
\end{lem}
\begin{proof}
Assume existence of a block diagonal matrix $\boldsymbol{W}$ and a matrix $\boldsymbol{X}$ such that $\boldsymbol{W}\boldsymbol{A}\boldsymbol{X}=\boldsymbol{H}$. This is equivalent to
$$\boldsymbol{W}_p \boldsymbol{A}_p \boldsymbol{X}= \boldsymbol{H}_p, \; 1\leq p \leq P.$$ 
Since, by assumption, $\mathrm{rank}(\boldsymbol{H}_p)=L$, it follows that $\mathrm{rank}(\boldsymbol{W}_p)=L \;\; \forall p$, making the matrix $\boldsymbol{W}$ invertible.

The reverse statement is trivial; if an invertible $\hat{\boldsymbol{W}}$ exists, then we can set $\boldsymbol{W}=\hat{\boldsymbol{W}}^{-1}.$
\end{proof}

For a randomly chosen $\boldsymbol{H}$, the condition $\mathrm{rank}(\boldsymbol{H}_p)=L$ holds with probability 1. We can then compute the WAX decomposition by invoking Lemma~\ref{lemeqv}, which yields the linear system
\begin{equation}\label{eq:AXWH}
    \boldsymbol{A}\boldsymbol{X}-\boldsymbol{W}^{-1}\boldsymbol{H}=\boldsymbol{0}_{M\times K}.
\end{equation}
Using the vectorization operator we get an equivalent linear system of equations
\begin{equation}\label{eq:lin_wax}
\boldsymbol{B}\boldsymbol{u}=\boldsymbol{0}_{MK\times 1},
\end{equation}
where $\boldsymbol{u}$ corresponds to the $(TK+ML)\times 1$ vector of unknowns,
\begin{equation}
\boldsymbol{u}=\begin{bmatrix}
\mathrm{vec}(\boldsymbol{X}) \\
\mathrm{vec}(\boldsymbol{W}_1) \\
\vdots \\
\mathrm{vec}(\boldsymbol{W}_P)
\end{bmatrix},
\end{equation}
and $\boldsymbol{B}$ is an $M K \times (TK+ML)$ block matrix of the form $\boldsymbol{B}=[\boldsymbol{B}_1 \;\boldsymbol{B}_2]$ resulting from the vectorization operation, with
\begin{equation} \label{eq:B1}
\boldsymbol{B}_1 = \mathbf{I}_K \otimes \boldsymbol{A}, \quad \boldsymbol{B}_2 = -(\boldsymbol{H}^\mathrm{T} \otimes \mathbf{I}_M)\mathbf{P}.
\end{equation}
$\mathbf{P}$ is an $M^2\times ML$ block matrix composed of identity matrices, $\boldsymbol{I}_{L}$, separated by rows of zeros so as to disregard the zeros in $\mathrm{vec}(\boldsymbol{W})$. The solution to \eqref{eq:lin_wax} can be found by setting $\boldsymbol{u}$ to be any vector in the null-space of $\boldsymbol{B}$, which will always be non-zero if condition \eqref{eq:cond_wax} is met (as will be seen shortly). Then, we can obtain the corresponding $\boldsymbol{W}^{-1}$ and $\boldsymbol{X}$ from $\boldsymbol{u}$ through inverse vectorization, and we should check that the resulting $\boldsymbol{W}^{-1}$ is full rank so that we can obtain $\boldsymbol{W}$ by taking the matrix inverse. Thus, the complexity of performing the WAX decomposition using the provided method is equivalent to that of finding the null-space of an $M K \times (TK+ML)$ matrix, and inverting $P$ matrices of dimension $L\times L$.

We have now set the ground for presenting our main result on the applicability of WAX decomposition, which is established in the following theorem.

\begin{thm} \label{WAXmain}
Assuming $N=L$, which implies that $P=M/L$ must evaluate to an integer value, fulfilling the inequality
\begin{equation}\label{eq:cond_wax}
T>\max\left(M\frac{K-L}{K},K-1\right)
\end{equation}
assures that, given a fixed randomly chosen $\boldsymbol{A}\in\mathbb{C}^{M\times T}$, a randomly chosen $\boldsymbol{H}\in\mathbb{C}^{M\times K}$ will admit a decomposition of the form \eqref{WAX} with probability 1.
\end{thm}
\begin{proof}
See Appendix A.
\end{proof}

An immediate result of Theorem~\ref{WAXmain} is that, since we are only interested in having $L\leq K$ due to its practical implications (otherwise there is no dimensionality reduction compared to local MF from Fig.~\ref{fig:arch_old} right side), then we have $T\geq L$ since $T\geq K$. The case where Theorem~\ref{WAXmain} is not fulfilled is considered in Section~\ref{section:lossy}.

For a randomly chosen $\boldsymbol{A}$, $\boldsymbol{W}^{-1}$ is full rank with probability 1, but note that some specific $\boldsymbol{A}$ matrices may lead to rank deficient $\boldsymbol{W}^{-1}$ even if \eqref{eq:cond_wax} is met. That is, for a poorly chosen matrix $\boldsymbol{A}$, the WAX decomposition of a matrix $\boldsymbol{H}$ cannot be performed. In what follows next we study conditions on $\boldsymbol{A}$ in order for the WAX decomposition to be feasible.

\subsection{Studying validity of matrix $\boldsymbol{A}$}
While Theorem~\ref{WAXmain} states that any randomly chosen $\boldsymbol{A}$ works for WAX decomposition, we are, from a practical perspective, interested in $\boldsymbol{A}$ matrices having simple forms (providing low computational complexity); for example, sparse matrices with elements in the set $\{0,1\}$. This would significantly simplify the combining network. However, for such a matrix, Theorem~\ref{WAXmain} no longer applies. Therefore, it is of importance to investigate the exceptions to Theorem~\ref{WAXmain}.\footnote{Note that, in the general case where $N\neq L$, selecting $\boldsymbol{A}$ as in Lemma~\ref{lem:N_L} maintains the overall sparsity properties of $\widetilde{\boldsymbol{A}}$ since the transformation $\boldsymbol{T}_{L,N}$ just replicates the matrix $\widetilde{\boldsymbol{A}}$ in different positions, filling the rest with 0s.}

\begin{defi}
We consider $\boldsymbol{A}$ to be valid for WAX decomposition \eqref{WAX} if the set of matrices $\boldsymbol{H}$ that does not admit such a decomposition has measure 0. This is equivalent as saying that the probability of a randomly chosen $\boldsymbol{H}$ admitting WAX decomposition for a valid $\boldsymbol{A}$ is 1.
\end{defi}

The following theorem will be useful to check if an $\boldsymbol{A}$ matrix is valid or not.
\begin{thm}\label{th:A_randH}
Consider a fixed $M\times T$ matrix $\boldsymbol{A}$, where $T$ fulfills \eqref{eq:cond_wax}, and a randomly chosen $\boldsymbol{H}\in \mathbb{C}^{M\times K}$, such that $\boldsymbol{B}=[\boldsymbol{B}_1\; \boldsymbol{B}_2]$ from \eqref{eq:B1} is full-rank. If the specific $\boldsymbol{H}$ admits WAX decomposition (\eqref{WAX} with $N=L$) for the given $\boldsymbol{A}$, then $\boldsymbol{A}$ will be valid for WAX decomposition with probability 1, i.e., any other randomly chosen $\boldsymbol{H}$ will admit WAX decomposition for the same $\boldsymbol{A}$ with probability 1.
\end{thm}
\begin{proof}
The proof is a side-result from the proof of Theorem~\ref{WAXmain} in Appendix A. We just have to notice that the determinant of $\widetilde{\boldsymbol{W}}_p$ still fulfills \eqref{eq:det_W} if $\boldsymbol{A}$ is fixed, so it is enough if we can find an $\boldsymbol{H}$ that gives $\mathrm{det}({\widetilde{\boldsymbol{W}}_p})\neq 0$ so that the determinant is 0 only for a countable set of $\boldsymbol{H}$ matrices.
\end{proof}

The main contribution of Theorem~\ref{th:A_randH} is that we can know if an $\boldsymbol{A}$ matrix is valid or not simply by trying to perform WAX decomposition of a single randomly chosen $\boldsymbol{H}$ using that $\boldsymbol{A}$. Theorem~\ref{th:A_randH} will be widely used for our simulation results since it is our main result on the validity of $\boldsymbol{A}$. An essentially equivalent statement to that of Theorem~\ref{th:A_randH} is to say that, for a given matrix $\boldsymbol{A}$, WAX decomposition of a randomly chosen matrix $\boldsymbol{H}$ is possible with probability either 1 or 0.

The following lemma states a necessary condition for $\boldsymbol{A}$ to be valid.
\begin{lem} \label{lem:blockrank}
Let $\widehat{\boldsymbol{A}}$ be an $F L \times T$ submatrix of $\boldsymbol{A}$ formed by selecting $F$ out of the $P$ blocks of size $L \times T$  that conform $\boldsymbol{A}$, $\{\boldsymbol{A}_1,\dots,\boldsymbol{A}_P\}$. If $\boldsymbol{A}$ is valid, then
\begin{equation}\label{eq:rank_A}
    \mathrm{rank}(\widehat{\boldsymbol{A}}) = \min(F L,K).
\end{equation}
\end{lem}
\begin{proof}
Let $\widehat{\boldsymbol{H}}$ be an $F L \times K$ submatrix of $\boldsymbol{H}$ formed by the $F$ respective $L\times K$ blocks of $\boldsymbol{H}$, and $\widehat{\boldsymbol{W}}$ an $F L\times F L$ block diagonal matrix with the $F$ respective $L \times L$ blocks from $\boldsymbol{W}$ as diagonal elements. If $\boldsymbol{A}$ is valid we can obtain $\widehat{\boldsymbol{W}}$ and $\boldsymbol{X}$ such that $\widehat{\boldsymbol{W}}\widehat{\boldsymbol{A}}\boldsymbol{X}=\widehat{\boldsymbol{H}}$ holds for any $\widehat{\boldsymbol{H}}$ except those in a set of measure 0. For a randomly chosen $\boldsymbol{H}$, $\mathrm{rank}(\widehat{\boldsymbol{H}})=\min(F L, K)$ with probability 1. Since $$ \mathrm{rank}(\widehat{\boldsymbol{W}}\widehat{\boldsymbol{A}}\boldsymbol{X})\leq \min \left(\mathrm{rank}(\widehat{\boldsymbol{W}}), \mathrm{rank}(\widehat{\boldsymbol{A}}), \mathrm{rank}(\boldsymbol{X})\right),$$
condition \eqref{eq:rank_A} must be fulfilled.
\end{proof}

A further necessary condition for $\boldsymbol{A}$ to be valid is given in Lemma~\ref{rankrow}.
\begin{lem} \label{rankrow}
\color{black}Let $\boldsymbol{A}_0$ be a submatrix of $\boldsymbol{A}$ formed by selecting $R$ rows from $\boldsymbol{A}$, where all rows are in different blocks $\boldsymbol{A}_p$. If $\boldsymbol{A}$ is valid, then
$$\mathrm{rank}(\boldsymbol{A}_0) > R\frac{K-L}{K}$$
\end{lem}
\begin{proof}
See Appendix B.
\end{proof}

We point out that, since $\boldsymbol{A}$ is an $M\times T$ matrix, $\mathrm{rank}(\boldsymbol{A}_0)$ cannot exceed $T$. However, with $T>M(K-L)/K$, it is guaranteed that $R\frac{K-L}{K}<T$.  

An immediate consequence of Lemma~\ref{rankrow} is that a block $\boldsymbol{A}_p$ cannot be repeated arbitrarily often in $\boldsymbol{A}$. In addition, any block $\boldsymbol{A}_p$ must have rank $L$ (see Lemma~\ref{lem:blockrank}). Repeating the block $\boldsymbol{A}_p$ $r$ times in $\boldsymbol{A}$, and selecting $\boldsymbol{A}_0$ as the same row within each of these $r$ blocks yields,
$$1>r\frac{K-L}{K},$$
which implies $r<\frac{K}{K-L}.$  Whenever $L\leq K/2$, $r=1$ so that each block $\boldsymbol{A}_p$ can only occur once in $\boldsymbol{A}$.

Appendix C includes another, less intuitive, necessary condition for $\boldsymbol{A}$ to be valid. Despite extensive efforts, we have not been able to establish sufficient conditions for having a valid $\boldsymbol{A}$. The provided necessary conditions might not be helpful in generating valid $\boldsymbol{A}$ matrices, but they constitute initial progress on the matter. Thus, the problem of establishing sufficient conditions for valid $\boldsymbol{A}$ matrices remains open.

\section{Finding Sparse $\boldsymbol{A}$ Matrices}
\label{section:sparse_A}
Let us keep restricting ourselves to the case where $N=L$ due to the generality of this case. Recall, however, that the transformation \eqref{eq:WAX_N}, which allows considering other $N$ values, maintains the overall sparsity properties of $\boldsymbol{A}$ (although the percentage of 1s can decrease since 0s are being padded). As we have mentioned previously, from an implementation point of view, it is desirable to have $\boldsymbol{A}$ as a sparse matrix of 1s and 0s with as few 1s as possible. The main reason is the direct relation between the number of 1s in $\boldsymbol{A}$ and the number of sum operations required to implement such a combining matrix. Recall that we can view $\boldsymbol{A}$ as the matrix associated to a predesigned hardware combining module, which could be intuitively implemented through a network of sum modules.

In the previous section, we provided some necessary conditions on matrix $\boldsymbol{A}$ for the WAX decomposition to be applicable. However, after extensive research on the matter, we have not been able to find sufficient conditions for having a valid $\boldsymbol{A}$. This motivates the need of simulation results to further understand the limits on the sparsity of $\boldsymbol{A}$. Nevertheless, we will support these simulations with a simple theoretical bound that allows us to gain better understanding of how sparse  $\boldsymbol{A}$ may be.

From Lemma~\ref{lem:blockrank}, we can say that each $L\times T$ block $\boldsymbol{A}_p$ $p \in \{1,\dots,P\}$, has to be of rank $L$. If we put it together with Lemma~\ref{rankrow} we can say that a valid $\boldsymbol{A}$ has a maximum of
\begin{equation}
    r_\mathrm{max} = \left\lceil \frac{K}{K-L}-1 \right\rceil
\end{equation}
equal $\boldsymbol{A}_p$ blocks. Therefore, if we aim at $\boldsymbol{A}$ matrices of 1s and 0s with the minimum number of 1s, we can lower bound the number of 1s through the following lemma.

\begin{lem}\label{lem:max1s}
Given a valid $\boldsymbol{A} \in \{ 1, 0\}^{M\times T}$, the number of 1s in $\boldsymbol{A}$, which incidentally coincides with its squared Frobenius norm, can be lower bounded by
\begin{equation}
    \Vert \boldsymbol{A}  \Vert_\mathrm{F}^2 \geq \sum_{k=1}^{Q-1} r_\mathrm{max}(k-Q) \binom{T}{k} +Q  M,
\end{equation}
where $Q$ is obtained by
\begin{equation}
    Q = \arg \min_q \sum_{k=1}^q r_\mathrm{max} \binom{T}{k} \geq M.
\end{equation}
\end{lem}
\begin{proof}
Let us impose the restriction that every row in $\boldsymbol{A}$ can be repeated a maximum of $r_\mathrm{max}$ times, as suggested by Lemma~\ref{rankrow}, where the restriction of selecting the rows from different blocks can be relaxed since we are interested in a lower bound. Each row of $\boldsymbol{A}$ must have at least a 1 at some location, otherwise, considering $L\leq T$, a row of only 0s would result in an $\boldsymbol{A}_p$ block with rank lower than $L$ (contradicting Lemma~\ref{lem:blockrank}). Therefore, we can bound the number of 1s in $\boldsymbol{A}$ by considering all possible combinations of rows having a single 1, each of which could be repeated a maximum of $r_\mathrm{max}$ times. Then, we can do the same for rows having 2 1s, and go on until we have enough rows to fill the $M$ rows of $\boldsymbol{A}$. Straightforward combinatorics result in the statements in the lemma.
\end{proof}

Even though the lower bound presented in Lemma~\ref{lem:max1s} might look a bit hard to compute at first sight, for reasonable values of the design variables ($M\leq 3 r_\mathrm{max}T$), it is enough to consider only rows with up to two 1s to have enough rows for filling $\boldsymbol{A}$, i.e., $Q\leq 2$. 

Figs. \ref{fig:opt_per1} and \ref{fig:opt_per1_L} show the minimum percentage of 1s required to have a valid $\boldsymbol{A}$ through numerical optimization, as well as the theoretical lower bound from Lemma~\ref{lem:max1s}, with respect to $100/T_\mathrm{opt}$. The optimum $T$, $T_\mathrm{opt}$, is computed as
\begin{equation}\label{eq:T_opt}
    T_\mathrm{opt}=\max\left(\left\lfloor M\frac{K-L}{K}+1\right \rfloor,K\right),
\end{equation}
which corresponds to the minimum integer $T$ that fulfills Theorem~\ref{WAXmain}. The reason for having $100/T_\mathrm{opt}$ is that it corresponds to a trivial lower bound in the percentage of 1s, i.e., if we have a single 1 per row. We can see that the relation between the minimum percentage of 1s and $100/T_\mathrm{opt}$ is close to linear. In Fig.~\ref{fig:opt_per1}, $K$ is fixed and $L$ is changed, while in Fig.~\ref{fig:opt_per1_L} it is the other way around. Although our initial intention was not to provide a tight bound to the numerical simulation, we can note that the obtained percentage of 1s for valid $\boldsymbol{A}$ matrices is in general close to the lower bound computed given by Lemma~\ref{lem:max1s} (within $\pm3\%$ error in the plots).

The algorithm for computing the minimum percentage of 1s consists of selecting $\boldsymbol{A}$ by adding 1s at random positions, with some constraints related to Lemma~\ref{lem:blockrank}, until a valid $\boldsymbol{A}$ has been found. Then, we reduce through exhaustive search the number of 1s in $\boldsymbol{A}$ as much as possible while maintaining its validity. The validity of $\boldsymbol{A}$ is evaluated by considering Theorem~\ref{th:A_randH}.

Table \ref{tab:n_ones} shows the minimum number of 1s for having a valid $\boldsymbol{A}$ considering different values of $K$ and $L$ (recall $L\leq K$ for information-lossless processing). The same algorithm as for Figs. \ref{fig:opt_per1} and \ref{fig:opt_per1_L} is employed, and $T$ is also selected to be $T_\mathrm{opt}$ from \eqref{eq:T_opt}. We can see that, even though from Figs. \ref{fig:opt_per1} and \ref{fig:opt_per1_L} the percentage of 1s increases with $L$ and decreases with $K$, for the total number of 1s it is the other way around, i.e., it increases with $K$ and decreases with $L$. This is because the size of $\boldsymbol{A}$ increases at a higher rate when $T$ is increased. However, if we implement $\boldsymbol{A}$ using sum modules with only two inputs, we would need to subtract $T$ to the number of 1s in $\boldsymbol{A}$, since a single 1 in a column of $\boldsymbol{A}$ would not require any sum module. Table \ref{tab:n_sums} shows the minimum number of 2-input sum modules required to have a valid $\boldsymbol{A}$, i.e., the values of this table correspond to the values of Table \ref{tab:n_ones} subtracting the respective $T$ value to each entry. As we can see, the number of sum modules required does not vary considerably with $K$ and $L$, and it remains in the order of $M$.
\begin{figure}[h]
	\centering
	\includegraphics[scale=0.55]{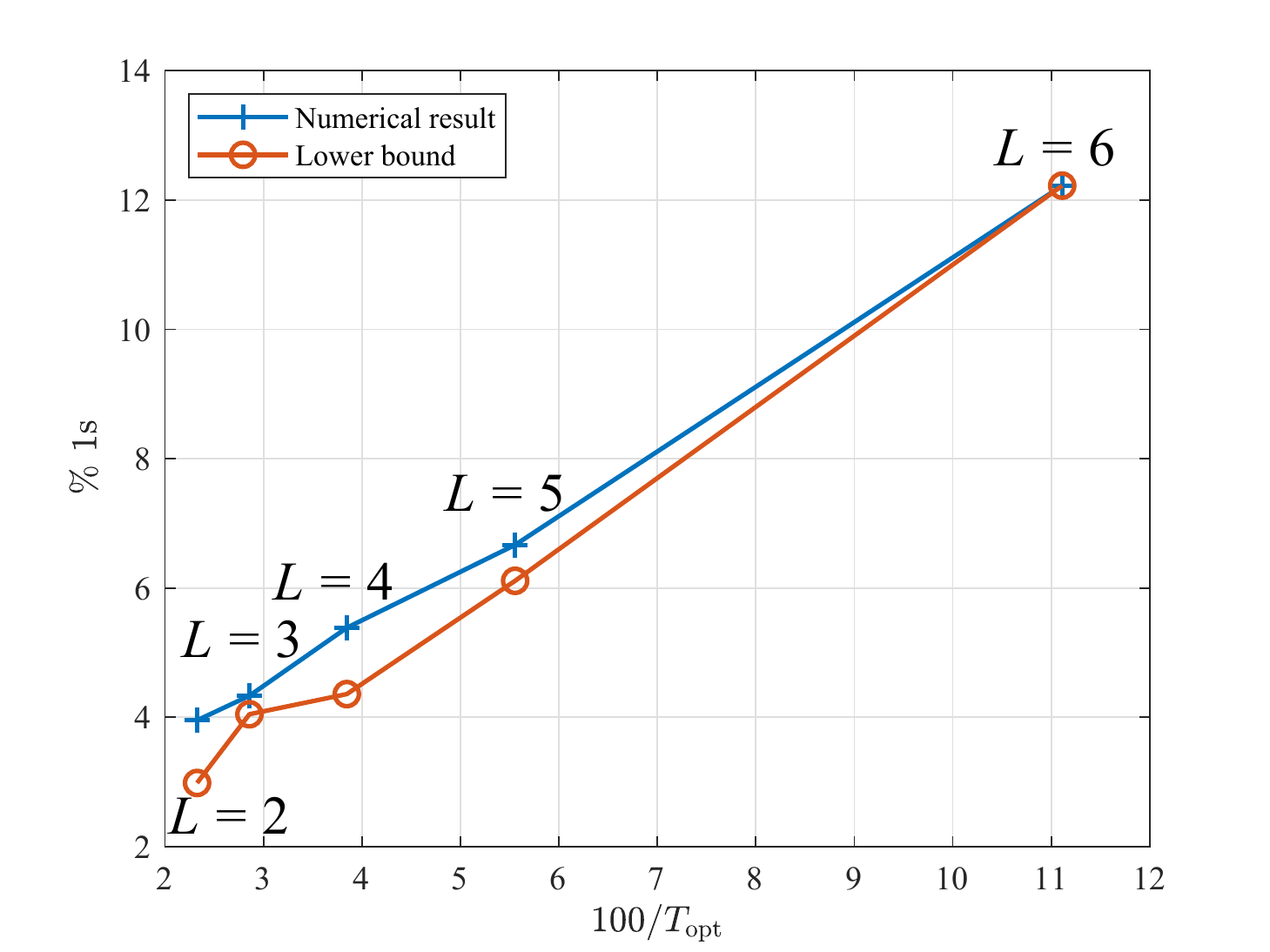}
\caption{Minimum percentage of 1s required for a valid $\boldsymbol{A}\in \{ 0,1 \}^{M\times T}$ with respect to $100/T_\mathrm{opt}$ for $M=60$, $K=7$, $L=\{2,3,4,5,6\}$ ($N=L$).}
	\label{fig:opt_per1}
\end{figure}

\begin{figure}[h]
	\centering
	\includegraphics[scale=0.55]{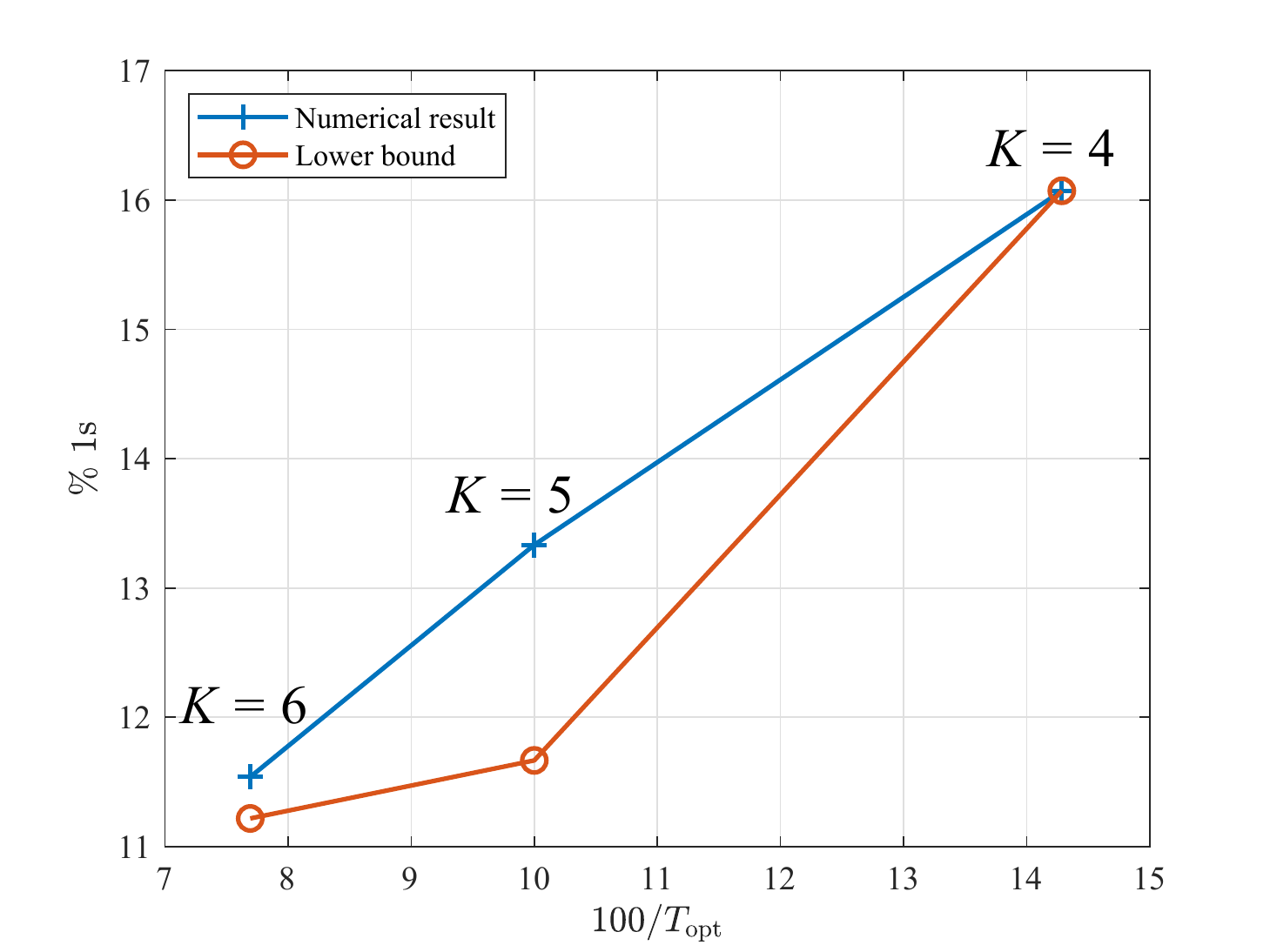}
\caption{Minimum percentage of 1s required for a valid $\boldsymbol{A}\in \{ 0,1 \}^{M\times T}$ with respect to $100/T_\mathrm{opt}$ for $M=24$, $K=\{6,5,4\}$, $L=3$ ($N=L$).}
	\label{fig:opt_per1_L}
\end{figure}

\begin{table}
    \centering
    \scalebox{1.3}{
    \begin{tabular}{|c||c|c|c|c|c|c|c|}
    \hline
         L/K & $1$ & $2$ & $3$ & $4$ & $5$ & $6$ & $7$  \\ \hhline{|==|=|=|=|=|=|=|}
         $1$ & 60 & 89 & 98 & 102 & 104 & 105 & 108 \\ \hline
         $2$ & - & 60 & 78 & 89 & 94 & 98 & 102 \\ \hline
         $3$ & - & - & 60 & 72 & 80 & 89 & 91 \\ \hline
         $4$ & - & - & - & 60 & 68 & 78 & 84 \\ \hline
         $5$ & - & - & - & - & 60 & 65 & 72 \\ \hline
         $6$ & - & - & - & - & - & 60 & 66 \\ \hline
    \end{tabular}}
    \caption{Minimum number of 1s required to have a valid $\boldsymbol{A}$ for $M=60$, $T=T_\mathrm{opt}$}
    \label{tab:n_ones}
\end{table}

\begin{table}
    \centering
    \scalebox{1.3}{
    \begin{tabular}{|c||c|c|c|c|c|c|c|}
    \hline
         L/K & $1$ & $2$ & $3$ & $4$ & $5$ & $6$ & $7$  \\ \hhline{|==|=|=|=|=|=|=|}
         $1$ & 59 & 58 & 57 & 56 & 55 & 54 & 56 \\ \hline
         $2$ & - & 58 & 57 & 58 & 57 & 57 & 59 \\ \hline
         $3$ & - & - & 57 & 56 & 55 & 58 & 56 \\ \hline
         $4$ & - & - & - & 56 & 55 & 57 & 58 \\ \hline
         $5$ & - & - & - & - & 55 & 54 & 54 \\ \hline
         $6$ & - & - & - & - & - & 54 & 57 \\ \hline
    \end{tabular}}
    \caption{Minimum number of 2-input sum modules required to implement a valid $\boldsymbol{A}$ for $M=60$, $T=T_\mathrm{opt}$}
    \label{tab:n_sums}
\end{table}

Fig.~\ref{fig:comp_cost} shows the percentage of valid $\boldsymbol{A}$ matrices for different percentage of 1s, where the 1s are placed at random positions with the only restriction that they have at least a 1 per row and per column. The algorithm employed consists of generating a big number of random $\boldsymbol{A}$ matrices with the corresponding percentage of 1s (while fulfilling a simple restriction related to Lemma~\ref{lem:blockrank} with $F=1$), and checking what percentage of them are valid (Theorem~\ref{th:A_randH} comes in handy here too). The different curves correspond to the same parameter combinations as Fig.~\ref{fig:opt_per1_L}, where we can find the minimum percentage of 1s that apply to the different values of $K$. As we can see, finding an $\boldsymbol{A}$ attaining a percentage of 1s slightly higher  than the optimum one (+10\%) is not that difficult since we can just put 1s at random positions and, with high probability (around 80\% in the case of $K=4$), this $\boldsymbol{A}$ will be valid.

\begin{figure}[h]
	\centering
	\includegraphics[scale=0.55]{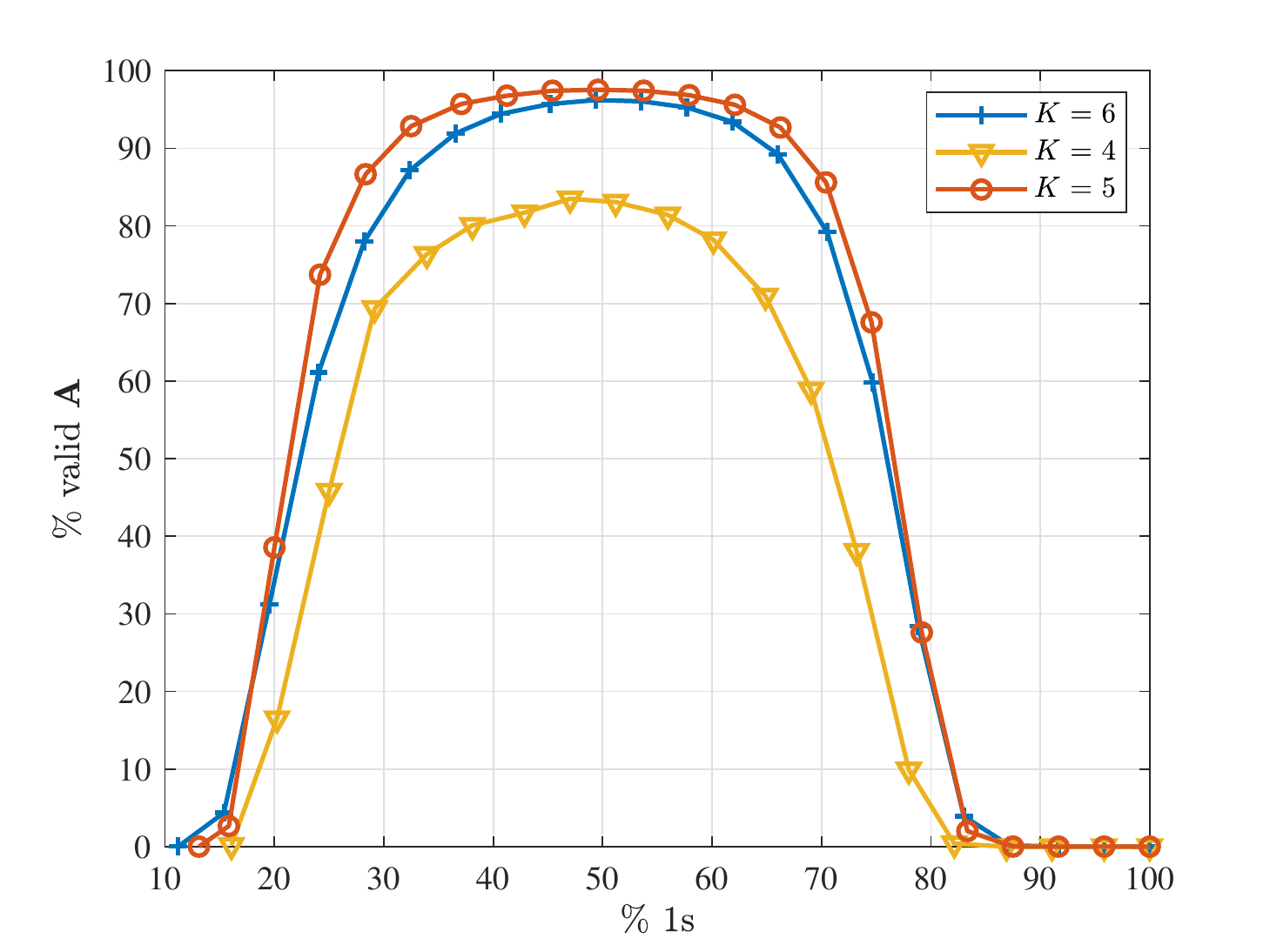}
\caption{Percentage of valid $\boldsymbol{A}\in \{0,1\}^{M\times T}$ with 1s at random positions with respect to the percentage of 1s for $M=24$, $L=3$ (N=L).}
	\label{fig:comp_cost}
\end{figure}

\section{Discussion and Examples} 
\label{section:tradeoff}
As mentioned previously, one of the goals of this paper is to find the trade-off between the different system parameters required for the equivalent processing to be information-lossless. An interesting case is when $N=1$, since it relates directly to the trade-off between the centralized and decentralized architectures shown in Fig.~\ref{fig:arch_old}. Assuming that our framework is equipped with a valid matrix $\boldsymbol{A}$, the trade-off between $L$ and $T$ comes directly from condition \eqref{eq:cond_wax}. We can select $T$ as $T_\mathrm{opt}$ from \eqref{eq:T_opt}, as can be seen in Fig.~\ref{fig:tradeoff}.

It is interesting to observe that we reach a reduction compared with the centralized architecture also for $L=1$. To elaborate a bit further, we observe that with $L=1$, the number of CPU inputs becomes $T=T_\mathrm{max} \triangleq \left\lfloor M-\frac{M}{K}+1\right \rfloor$. This reduction comes about since we have allowed the antennas to perform multiplications, which leads to a reduction in the number of CPU inputs from $M$ to, at most, $T_\mathrm{max}$. The centralized architecture, illustrated in the left part of Fig.~\ref{fig:arch_old}, has the same number of outputs per antenna, namely $1$, but does not perform any multiplications. Therefore, the CPU must operate with $T=M$. 
If we let $L_{\mathrm{mult}}$ denote the number of multiplications per antenna, the centralized architecture corresponds to $L_{\mathrm{mult}}=0$, and we can select $T$ as
$$T= \left\{ \begin{array}{ll} M & L_{\mathrm{mult}} =0 \\ \max\left(\left\lfloor M\frac{K-L_{\mathrm{mult}}}{K}+1\right \rfloor,K\right) & L_{\mathrm{mult}} >0. \end{array} \right.$$
This is conceptually illustrated in the right part of Fig.~\ref{fig:tradeoff}.
\begin{figure}
     \centering
     \begin{subfigure}[b]{0.23\textwidth}
         \centering
         \includegraphics[scale=0.34]{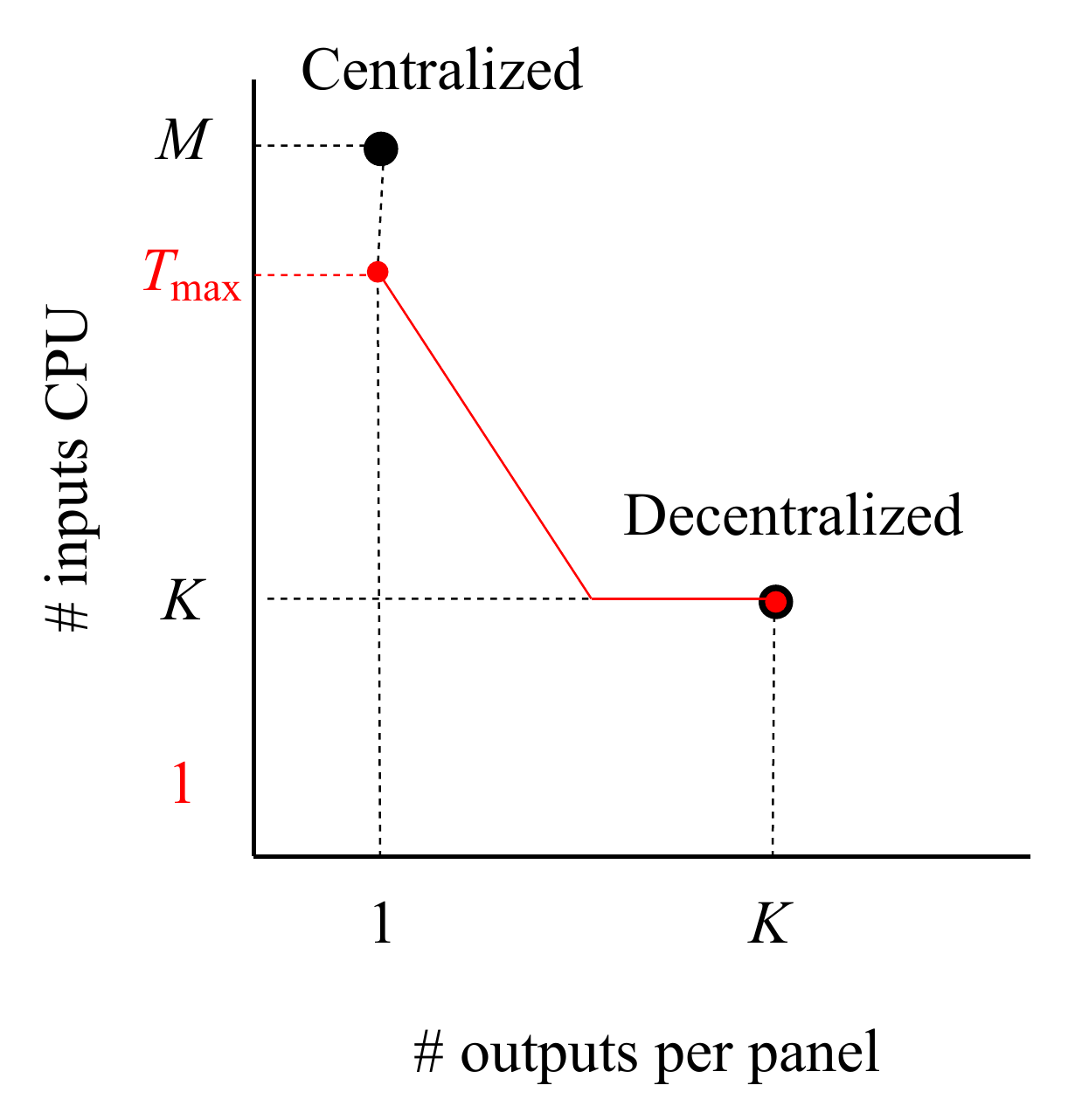}
     \end{subfigure}
     \hfill
     \begin{subfigure}[b]{0.23\textwidth}
         \centering
         \includegraphics[scale=0.34]{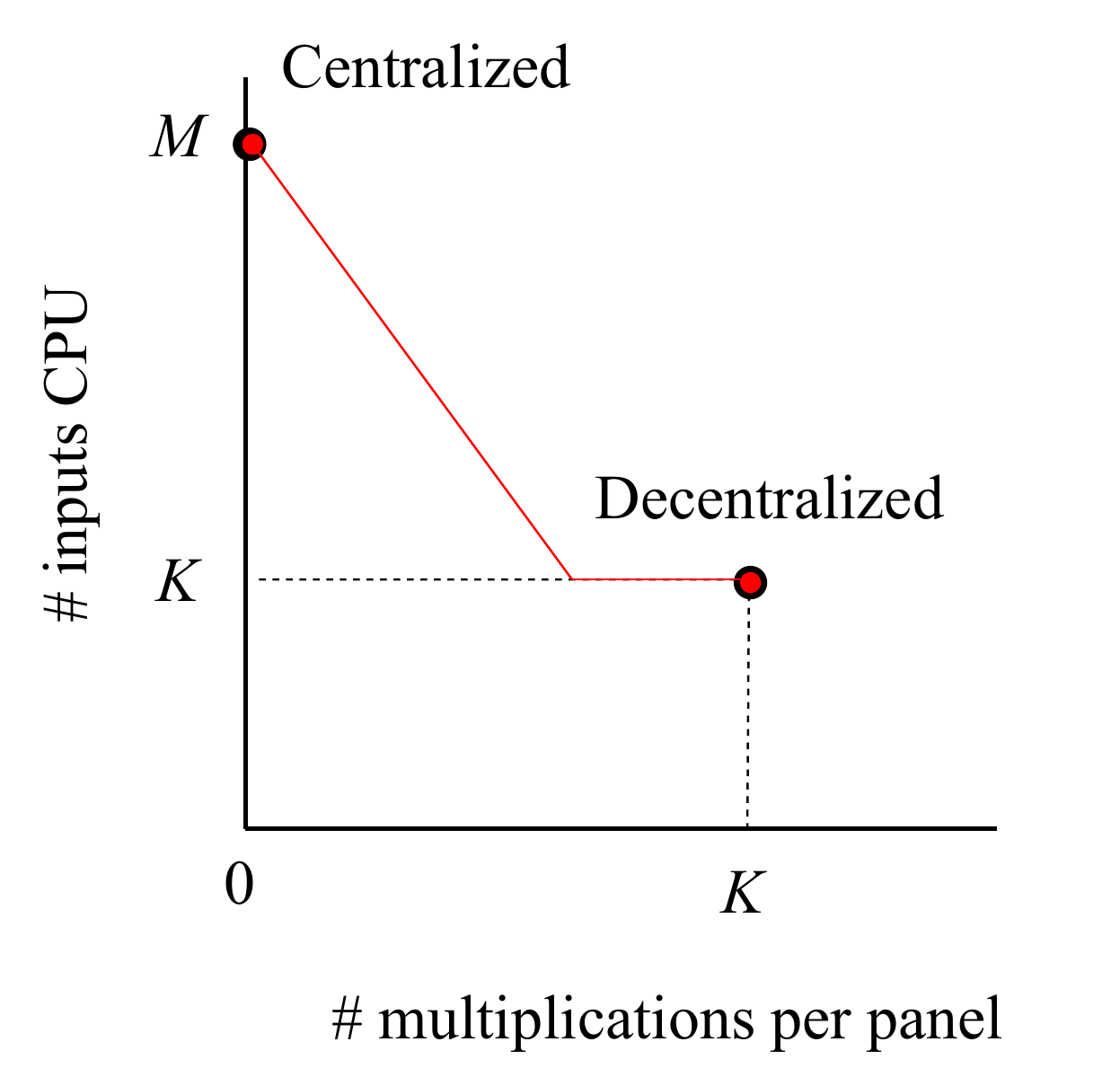}
     \end{subfigure}
        \caption{Number of inputs to the CPU (related to $T$) v.s. number of outputs (left)/number of multiplications (right) per antenna (related to $L$).}
        \label{fig:tradeoff}
\end{figure}

We conclude this section with a few examples.
\begin{ex} 
Assume a design with a CPU limited to $T~\leq~50$ inputs, and antenna panels with $L=2$ outputs. We now consider how many antennas and users ($M$ and $K$, respectively) can be handled by the system. From \eqref{eq:cond_wax}, we have  $50>M\frac{K-2}{K}$ implying that
$$50\frac{K}{K-2}>M.$$
To maximize the left hand side, excluding the special case $K=L=2$ (which allows for an unlimited number of antennas), we set $K=3$ and obtain $M<150$ so that we can at most use $149$ antennas. Differently put, if we choose to equip the base station with 149 antennas, we can at most serve $K=3$ users. With 150+ antennas, only 2 users can be served. Setting $K=4$, yields that at most 99 antennas can be used.
\end{ex}

We next provide two numerical examples of the WAX decomposition. The first one is meant to illustrate that it is indeed possible to obtain valid sparse matrices $\boldsymbol{A}$ comprising only elements in the set $\{0,1\}$.
\begin{ex} Let $M=100$, $N=4$, $P=25$, $K=10$, and $L=4$. From Theorem~\ref{WAXmain}, we have that $T>100\times 0.6=60$, so we take $T=61$. It can be numerically verified that the matrix
\begin{equation} \label{ex1} 
\boldsymbol{A} = \left[\begin{array}{c} 
\boldsymbol{I}_{61} \\
\hline
\begin{array}{c|c}
\boldsymbol{I}_{39} & \begin{array}{c} \boldsymbol{I}_{22} \\ \hline \begin{array}{c|c} 
\boldsymbol{I}_{17} & \begin{array}{c} \boldsymbol{I}_5 \\ \boldsymbol{I}_5 \\ \boldsymbol{I}_5 \\ \hline \begin{array}{ccc} \boldsymbol{I}_2 & \boldsymbol{I}_2 & \boldsymbol{0}_{2\times 1} \end{array}
\end{array}
\end{array} 
\end{array}
\end{array}
\end{array}
\right]\end{equation}
is valid. We designed this $\boldsymbol{A}$ by aiming at a minimum number of non-zero elements, while satisfying both Lemma~\ref{lem:blockrank} and Lemma~\ref{rankrow}. It can be verified that $\boldsymbol{A}$ has 158 ones and 5942 zeros. Thus, merely 2.6\% of $\boldsymbol{A}$ is non-zero. 
\end{ex}

Our next example is providing the reader with a graphical illustration of the WAX decomposition.
\begin{ex} 
Let $M=8$, $N=2$, $P=4$, $K=5$, and $L=2$. Thus, $T>4.8$, so we select $T=5$. In this case, the number of variables  and the number of equations associated to the linear system \eqref{eq:AXWH} is $TK+ML=41$ and $MK=40$, respectively; thus, we have precisely one more variable than equations. A particular example of the WAX decomposition for the given parameters is shown in \eqref{ex2}. The strength of the WAX decomposition is that, for any $\boldsymbol{H}$, except for those in a set of measure 0, the matrix $\boldsymbol{A}$ can be kept as it is while only $\boldsymbol{W}$ and $\boldsymbol{X}$ need to change.
\begin{figure*}[b]
\hrulefill
$$ $$
\begin{equation} \label{ex2}
\underbrace{\left[\begin{BMAT}{cccccccc}{cccccccc} \,1&-1&0&0&0&0&0&0 \\
                                          \,2&1&0&0&0&0&0&0 \\
                                       \,   0&0&1&1&0&0&0&0 \\
                                       \,   0&0&0&-1&0&0&0&0 \\
                                      \,    0&0&0&0&2&1&0&0 \\
                                      \,    0&0&0&0&2&-1&0&0 \\
                                      \,    0&0&0&0&0&0&2&-1 \\
                                      \,    0&0&0&0&0&0&2&2
\addpath{(0,6,1)rruulldd}
\addpath{(2,4,1)rruulldd}
\addpath{(4,2,1)rruulldd}
\addpath{(6,0,1)rruulldd}
\end{BMAT} \right]}_{\boldsymbol{W}}
\underbrace{\left[ \begin{BMAT}{ccccc}{cccccccc}1&0&0&0&0 \\
                                          0&1&0&0&0 \\
                                          0&0&1&0&0 \\
                                          0&0&0&1&0 \\
                                          0&0&0&0&1 \\
                                          1&0&0&1&1 \\
                                          0&1&0&1&1 \\
                                          0&0&1&1&1
\end{BMAT} \right]}_{\boldsymbol{A}} 
\underbrace{\left[ \begin{BMAT}{ccccc}{ccccc} 
-2&-1&-1&2&2 \\
1&-2&-1&1&2 \\
1&-1&-2&-1&-2 \\
0&2&0&-1&0 \\
0&-1&2&1&-2  
\end{BMAT} \right]}_{\boldsymbol{X}} =
\underbrace{\left[ \begin{BMAT}{ccccc}{cccccccc} 
-3&1&0&1&0 \\
-2&-4&-3&5&6 \\
1&1&-2&-2&-2 \\
0&-2&0&1&0 \\
-2&-2&5&4&-4 \\
2&-2& 3&0&-4 \\
1&-2&2&3&4 \\
4&-2&2&0&-8 
\end{BMAT} \right]}_{\boldsymbol{H}}
\end{equation}
\end{figure*}
\end{ex}

%
%

\section{Information-loss without WAX decomposition}
\label{section:lossy}
Throughout the previous sections of the paper we have focused on performing information-lossless processing within our framework, i.e., we have focused on the cases where the maximization \eqref{eq:maximize} leads to $I_{\boldsymbol{Z},\boldsymbol{S}}(\boldsymbol{z}; \boldsymbol{s})=I_{\boldsymbol{Y},\boldsymbol{S}}(\boldsymbol{y}; \boldsymbol{s})$. We have defined the WAX decomposition, which allows performing information-lossless processing within our framework. Theorem~\ref{WAXmain} sets the main constraints on the system dimensions for WAX decomposition, and equivalently information-lossless processing, to be possible within our framework. However, it is of great interest to know the information-loss produced when  WAX decomposition is not possible, i.e., we would like to solve \eqref{eq:maximize} when Theorem~\ref{WAXmain} is not satisfied.

Solving \eqref{eq:maximize} is a research challenge in itself which might lead to future work on the topic. In this section we will present initial ideas, as well as numerical results using standard optimization methods, to get an overall understanding of the information-losses that are induced when having lower $T$ than the minimum required from Theorem~\ref{WAXmain}. We will again focus on the general case where $N=L$, and we will only consider randomly chosen $\boldsymbol{A}$ and $\boldsymbol{H}$ matrices.

\subsection{Approximate MF}

A simple first approach, intuitively related to how we compute WAX decomposition, is to work on the minimization problem
\begin{equation}\label{eq:min_norm}
\begin{aligned}
    &\underset{\boldsymbol{X},\{\hat{\boldsymbol{W}}_p\}_{p=1}^P}{ \text{minimize}} & \; \Vert \boldsymbol{A} \boldsymbol{X} - \hat{\boldsymbol{W}}\boldsymbol{H} \Vert^2_{\mathrm{F}}, \\
    &\;\;\;\;\;\;\;\;\;\;\; \mathrm{s.t.} & \mathrm{rank}(\hat{\boldsymbol{W}}_p)=L, \, \forall p \\
    & & \Vert \boldsymbol{X} \Vert + \Vert \hat{\boldsymbol{W}} \Vert = c,
\end{aligned}
\end{equation}
where $\hat{\boldsymbol{W}}$ has the same structure as \eqref{eq:w} for $N=L$ and $P=M/L$, and the actual matrix $\boldsymbol{W}$ to be used in our framework would be obtained as $\boldsymbol{W}=\hat{\boldsymbol{W}}^{-1}$. The last constraint in \eqref{eq:min_norm} ensures a non-zero solution, where the scalar $c$ can be any non-zero real value.\footnote{Note that, when applying $\boldsymbol{W}=\hat{\boldsymbol{W}}^{-1}$, the multiplication associated to \eqref{eq:z_proc} will cancel out any common scaling of $\hat{\boldsymbol{W}}$ and $\boldsymbol{X}$.} The minimization \eqref{eq:min_norm} leads to the WAX decomposition when Theorem~\ref{WAXmain} applies. This means that we could solve both problems using the same approach, and thus, without altering the overall complexity. In case \eqref{eq:cond_wax} is fulfilled, \eqref{eq:min_norm} would give 0, and the solution would also solve the maximization \eqref{eq:maximize}, which is not true in general. 

Solving \eqref{eq:min_norm} can be seen as applying approximate MF within our framework. This minimization can be found in closed-form when $L\leq \min(T,K) $ as we will now prove. Let us rewrite the norm as
\begin{equation}
    \Vert \boldsymbol{A} \boldsymbol{X} - \hat{\boldsymbol{W}}\boldsymbol{H} \Vert^2_{\mathrm{F}} =\sum_{p=1}^P \Vert \boldsymbol{A}_p \boldsymbol{X} - \hat{\boldsymbol{W}}_p\boldsymbol{H}_p \Vert^2_{\mathrm{F}}.
\end{equation}
Assuming the optimum $\boldsymbol{X}$ has been fixed, we would have $\hat{\boldsymbol{W}}_p=\boldsymbol{A}_p \boldsymbol{X}\boldsymbol{H}_p^\dagger$ ($\boldsymbol{H}_p^\dagger$ being the right pseudo-inverse of $\boldsymbol{H}_p$), which is of rank $\min(T,K,L)$ when $\boldsymbol{A}$ and $\boldsymbol{H}$ are randomly chosen. We can restrict ourselves to $L \leq \min(T,K)$ due to its previously mentioned practical interest. In this case, \eqref{eq:min_norm} is solved by considering the equivalent linear system from \eqref{eq:lin_wax}
\begin{equation}
    \Vert \boldsymbol{A} \boldsymbol{X} - \hat{\boldsymbol{W}}\boldsymbol{H} \Vert^2_{\mathrm{F}} = \boldsymbol{u}^H \boldsymbol{B}^H \boldsymbol{B} \boldsymbol{u}.
\end{equation}
The vector $\boldsymbol{u}$ corresponds to the entries of matrices $\boldsymbol{X}$ and $\hat{\boldsymbol{W}}$, so the constraint on the norms of $\hat{\boldsymbol{W}}$ and $\boldsymbol{X}$ corresponds to an arbitrary norm constraint on $\boldsymbol{u}$. With this in mind, the 
solution to the minimization problem is obtained by setting $\boldsymbol{u}$ to be the eigenvector associated to the lowest eigenvalue of $\boldsymbol{B}^H \boldsymbol{B}$. If the conditions for WAX were fulfilled, the lowest eigenvalue of $\boldsymbol{B}^H \boldsymbol{B}$ would be $0$, which would lead to the information-lossless solution.

\subsection{Antenna selection}
Another practical approach that would give a lower bound to the maximum $I_{\boldsymbol{Z},\boldsymbol{S}}(z,s)$ from \eqref{eq:maximize} is to consider antenna selection so as to reduce $M$ until Theorem~\ref{WAXmain} is satisfied. Note that this approach is valid only when $T\geq K$, which was not the case with the previous approach. Furthermore, for ${N>1}$, the antenna selection would correspond to panel selection since we need an integer number of panels, which means we can only reduce $N$ antennas at a time.

How to optimally make the panel selection is also a research problem in itself, but we limit our results to simple selection where the panels having the highest local channel matrix norms are the ones being used.

\subsection{Numerical results}
\begin{defi}Let us define the relative rate for our framework as
\begin{equation}
\mathbb{E}_{\boldsymbol{A},\boldsymbol{H}} \left( \frac{I_{\boldsymbol{Z},\boldsymbol{S}}(z,s)}{I_{\boldsymbol{Y},\boldsymbol{S}}(y,s)}\right),
\end{equation}
where $\boldsymbol{A}$ and $\boldsymbol{H}$ are standard IID Gaussian random matrices.
\end{defi}

Figs. \ref{fig:rate_loss} and \ref{fig:rate_loss_L} show the average relative rates, which have been optimized using norm minimization \eqref{eq:min_norm} and panel selection. We have also included as a comparison the average relative rate obtained by standard brute force numerical optimization using the result from \eqref{eq:min_norm} as starting point. This numerical optimization can give us a hint of what average relative rates could be achieved if more clever optimization approaches were considered. In the plots, we have considered the range of $L$ and $T$ values not accepting a WAX decomposition according to Theorem~\ref{WAXmain}. Fig.~\ref{fig:rate_loss} shows that when we reduce $T$ below $T_\mathrm{opt}$ the performance gets slowly degraded, so the system would still be able to work at acceptable rates even if we cannot perform WAX decomposition. In the case of Fig.~\ref{fig:rate_loss_L}, reducing $L$ below $L_\mathrm{opt}$,\footnote{ $L_\mathrm{opt}$ can be obtained from Theorem~\ref{WAXmain} as a converse of $T_\mathrm{opt}$.} attains a steeper loss, but the degradation is still reasonable. We can also see that the norm minimization associated to \eqref{eq:min_norm} presents a considerable loss with respect to the numerical optimization, but it can still serve as a simple auxiliary method that allows our framework to keep serving users under conditions where WAX is not possible, e.g., if the number of users increase. In fact, \eqref{eq:min_norm} corresponds to the WAX decomposition when the parameters allow it, so that no further processing would be needed in such a system. 

Panel selection, can perform better than the norm minimization in some cases, but the limitation of having to disregard full panels makes it perform slightly worse for most values of $T$ in Fig.~\ref{fig:rate_loss}. Furthermore, this method is slightly less versatile than norm minimization since it is not available for all possible values of $T$, as can be seen in Fig.~\ref{fig:rate_loss}. However, we can note that $T< K$ intuitively translates into having to discard the data from certain users, which could potentially be considered in more complex antenna selection schemes.

\begin{figure}[h]
	\centering
	\includegraphics[scale=0.55]{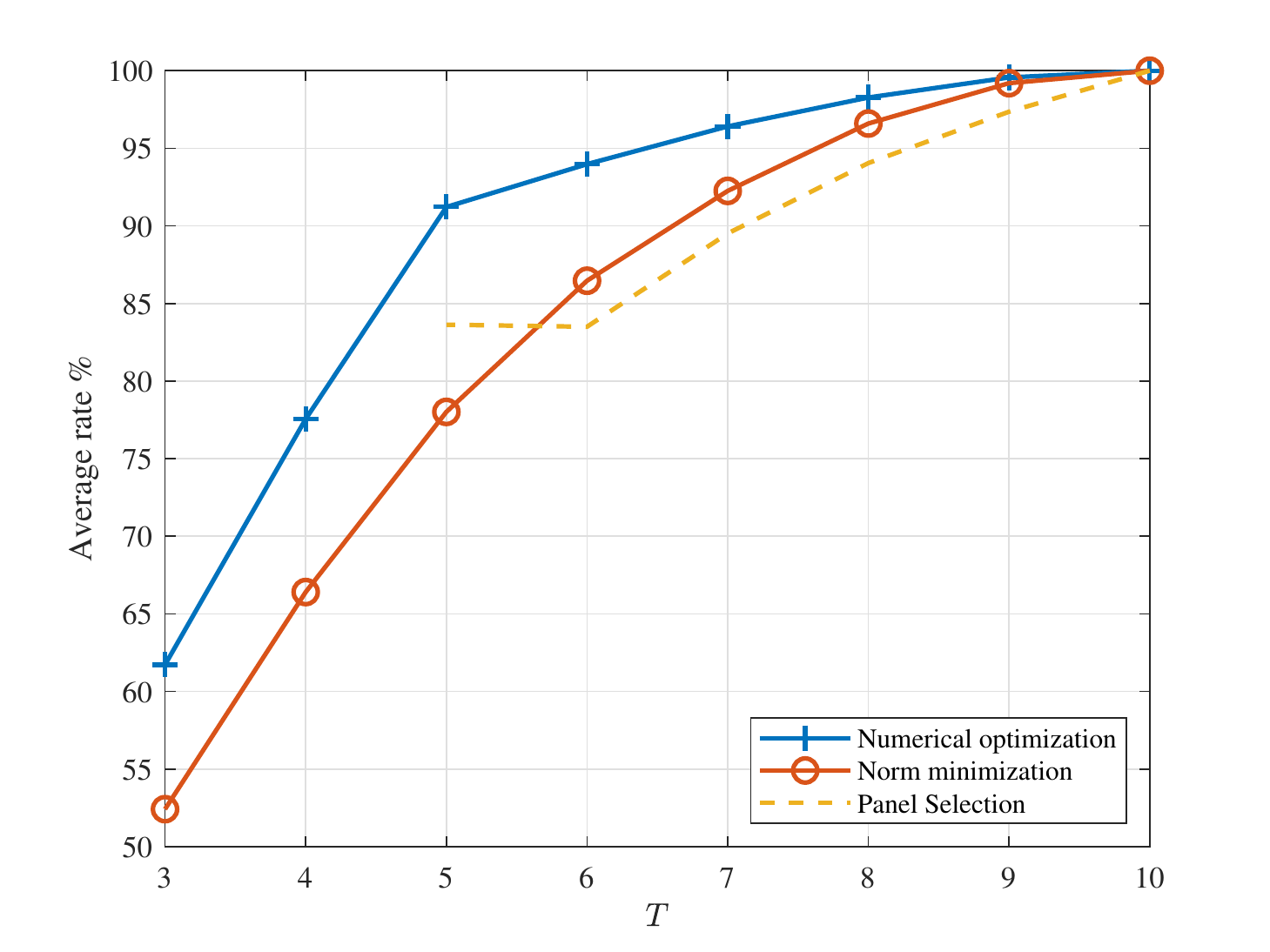}
\caption{Average relative rate (in percentage), $I_{\boldsymbol{Z},\boldsymbol{S}}(z,s)/I_{\boldsymbol{Y},\boldsymbol{S}}(y,z)$, with $\boldsymbol{A}$ and $\boldsymbol{H}$ having IID Gaussian entries, for $M=24$, $K=5$, $L=3$, $T_\mathrm{opt}=10$.}
	\label{fig:rate_loss}
\end{figure}

\begin{figure}[h]
	\centering
	\includegraphics[scale=0.55]{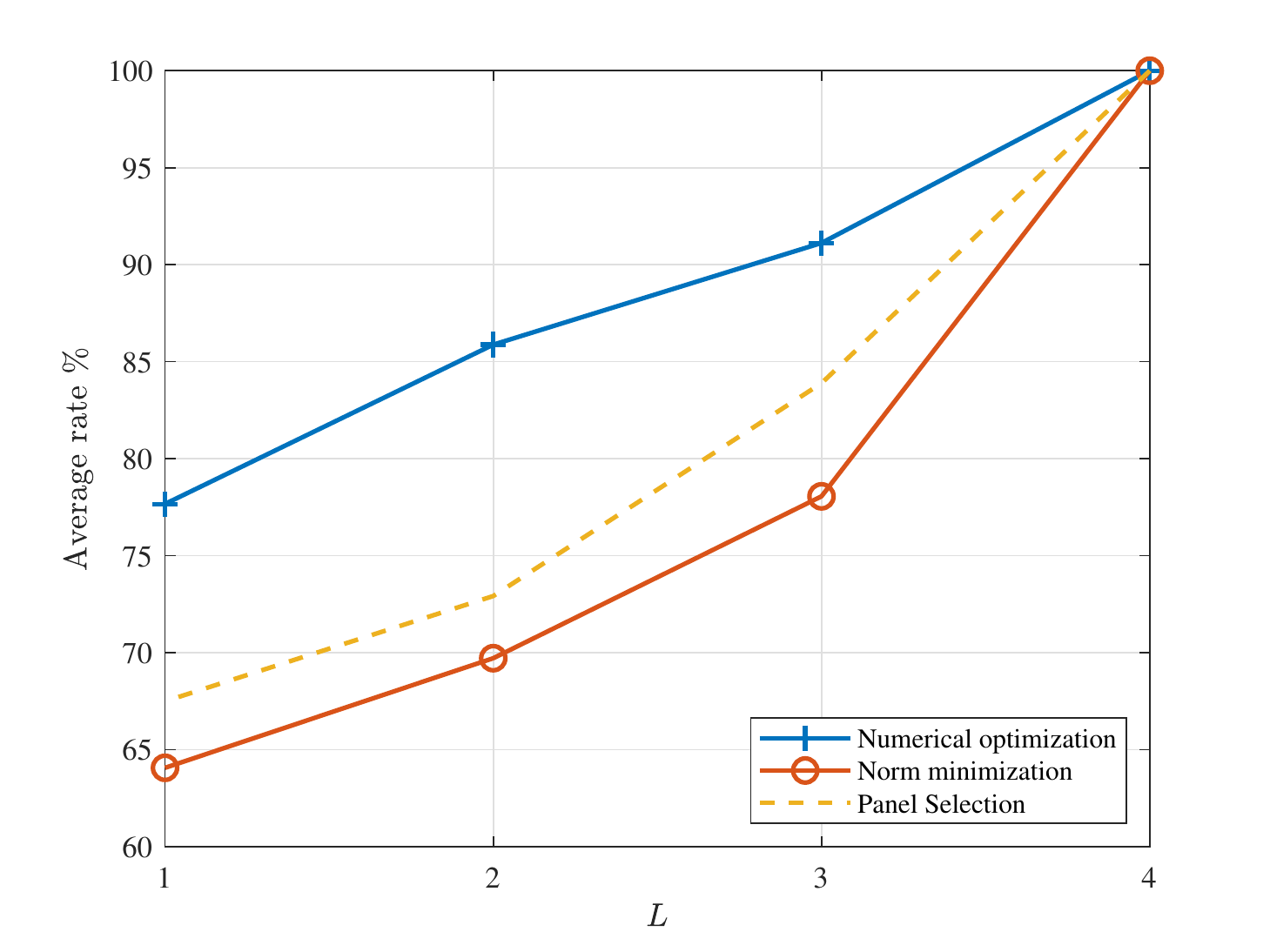}
\caption{Average relative rate (in percentage), $I_{\boldsymbol{Z},\boldsymbol{S}}(z,s)/I_{\boldsymbol{Y},\boldsymbol{S}}(y,z)$, with $\boldsymbol{A}$ and $\boldsymbol{H}$ having IID Gaussian entries, for $M=24$, $K=5$, $T=5$, $L_\mathrm{opt}=4$.}
	\label{fig:rate_loss_L}
\end{figure}

\section{Conclusions}
\label{section:conclusions}
We have introduced a general framework that allows the exploitation of the trade-off between complexity (number of multiplications/outputs per panel) and level of decentralization (connections to CPU) in multi-antenna architectures. We have presented the WAX decomposition, a matrix decomposition that achieves information-lossless processing within our framework under some restrictions. Said restrictions also describe the trade-off between number of multiplications/outputs per panel and number of connections to the CPU if we consider only information-lossless processing. Furthermore, we have studied the problem of finding simple combining networks ($\boldsymbol{A}$ matrices) that admit WAX decomposition within our framework. Finally, we have broadly studied the information-loss produced when the system parameters lead to non-availability of the WAX decomposition. 

Future work could include a deeper study on the cases where WAX decomposition is not available and information-lossy processing has to be applied. This same line would include the design of efficient algorithms and analytical solutions to optimize the achievable information rates. Furthermore, finding sufficient conditions for the combining network (i.e., $\boldsymbol{A}$) to be valid for WAX decomposition still remains unresolved. Other lines of work could include narrowing the study to some of the exceptions that we have not covered such as when the parameters are not divisible, etc.

\section*{Appendix A: Proof of Theorem~\ref{WAXmain}}
We first make the observation that the rank of $\boldsymbol{A}$ cannot be lower than the rank of $\boldsymbol{H}$. The rank of a randomly chosen $\boldsymbol{A}$ is $\min(M,T)$ with probability 1. Assuming that $M\geq K$, this implies  $T\geq K$, expressed as $T> \max(\cdot, K-1)$ in the statement.

We next provide  a lemma that will be useful.
\begin{lem} \label{lemma_t1}
Let $\boldsymbol{W}$ and $\hat{\boldsymbol{W}}$ be two matrices of the same form as $\boldsymbol{W}$ in \eqref{eq:w} with $N=L$ (i.e., they are block diagonal matrices where the blocks are $L\times L$ matrices). If $\boldsymbol{AX}=\hat{\boldsymbol{W}} \boldsymbol{H}$ is solvable such that $\mathrm{det}(\hat{\boldsymbol{W}}_p) > 0, \; 1\leq p \leq P$, then $\boldsymbol{W}\boldsymbol{A}\boldsymbol{X}=\boldsymbol{H}$ is solvable.
\end{lem}
\begin{proof}
Suppose $\boldsymbol{A}\boldsymbol{X}=\hat{\boldsymbol{W}}^H \boldsymbol{H}$ is solvable such that $\mathrm{det}(\hat{\boldsymbol{W}}_p) > 0, \; 1\leq p \leq P$. This implies that $\hat{\boldsymbol{W}}^{-1}$ exists. Thus,
$$\hat{\boldsymbol{W}}^{-1}{\boldsymbol{A}}{\boldsymbol{X}}=\boldsymbol{H}.$$
The lemma follows by observing that $\hat{\boldsymbol{W}}^{-1}$ is of the same form as $\boldsymbol{W}$, so we can take $\boldsymbol{W}=\hat{\boldsymbol{W}}^{-1}$.
\end{proof}

Let us now study $\boldsymbol{AX}=\hat{\boldsymbol{W}} \boldsymbol{H}$. Said matrix equation specifies $MT$ linear equations in $TK+ML$ variables, and hence, it is solvable if $T>M(K-L)/K$, as can be seen from the vectorized version in \eqref{eq:lin_wax}. It remains to show that for randomly chosen $\boldsymbol{A}$ and $\boldsymbol{H}$, the solution satisfies $\mathrm{det}(\hat{\boldsymbol{W}}_p) > 0, \; 1\leq p \leq P$. Let us define $\mathcal{V}$ as the set of admissible solutions, i.e., 
\begin{equation}\begin{aligned}
\mathcal{V}=\{\boldsymbol{A}, \boldsymbol{H} \, | &\, \exists \hat{\boldsymbol{W}}, \boldsymbol{X}: \boldsymbol{AX}=\hat{\boldsymbol{W}} \boldsymbol{H},\\
&\mathrm{det}(\hat{\boldsymbol{W}}_p)\neq 0, \, \forall p, \mathrm{det}(\boldsymbol{B}\boldsymbol{B}^H)\neq 0\},
\end{aligned}
\end{equation}
where $\boldsymbol{B}=[\boldsymbol{B}_1 \boldsymbol{B}_2]$ is again the matrix associated to the equivalent linear system \eqref{eq:lin_wax}, which is given by \eqref{eq:B1}. Assuming that $\boldsymbol{B}$ is full-rank, the solution to $\boldsymbol{AX}=\hat{\boldsymbol{W}} \boldsymbol{H}$ depends on $F=TK+ML-MK$ free variables, here denoted by $\{z_f\}$. 
The solution $\{\hat{\boldsymbol{W}}\}_{i,j}$ is a linear combination of the free variables $\{z_f\}$ where the weights depend on $\boldsymbol{A}$ and $\boldsymbol{H}$, i.e.,
$$\{\hat{\boldsymbol{W}}_p\}_{i,j}=\sum_{f=1}^F c_{p,i,j,f}(\boldsymbol{A},\boldsymbol{H})z_f.$$
Note that the number of free variables, $F$, can be increased if $\boldsymbol{B}$ is not full-rank (leading to a different polynomial expression), that is why we are only interested in the solutions giving a full-rank $\boldsymbol{B}$, i.e., $\mathrm{det}(\boldsymbol{B}\boldsymbol{B}^H)\neq 0$. Thus, the following lemma will come in handy.

\begin{lem} \label{lem:det_B}
Given randomly chosen matrices $\boldsymbol{H}$ and $\boldsymbol{A}$, the matrix $\boldsymbol{B}=[\boldsymbol{B}_1 \boldsymbol{B}_2]$ ($\boldsymbol{B}_1$ and $\boldsymbol{B}_2$ are defined as in \eqref{eq:B1}) fulfills $\mathrm{det}(\boldsymbol{B}\boldsymbol{B}^H)\neq 0$ with probability 1.
\end{lem}
\begin{proof}
We can define the determinant $\mathrm{det}(\boldsymbol{B}\boldsymbol{B}^H)$ as a polynomial expression of the form
$$ \mathrm{det}(\boldsymbol{B}\boldsymbol{B}^H) = \sum_g b_g(\boldsymbol{A},\boldsymbol{H}) \prod_{i,j,l,k} \{ \boldsymbol{H}\}_{i,j}^{h_{i,j,g}} \{ \boldsymbol{A}\}_{l,k}^{a_{l,k,g}}.
$$
This polynomial expression will evaluate to 0 only for a countable set (which thus attains probability 0) of $\boldsymbol{A}$ and $\boldsymbol{H}$ matrices if we can find at least an $\boldsymbol{A}$ and an $\boldsymbol{H}$ such that $\mathrm{det}(\boldsymbol{B}\boldsymbol{B}^H) \neq 0$. One example of this is when $\boldsymbol{A}$ is chosen randomly and we have $\{\boldsymbol{H} \}_{i,j}=\{\boldsymbol{A}\}_{i,j}$, $i\in \{ 1,\dots M \}$, $j \in \{ 1,\dots, K\}$.
\end{proof}

A similar reasoning to the one used in the proof of Lemma~\ref{lem:det_B} can be applied to show that for randomly chosen $\boldsymbol{A}$ and $\boldsymbol{H}$ we have $\mathrm{\hat{\boldsymbol{W}}}\neq 0$ with probability 1. The determinant $\mathrm{det}(\hat{\boldsymbol{W}}_p)$ can be written as a polynomial combination of the previous $\{\hat{\boldsymbol{W}}_p\}_{i,j}$. Thus we can express it as
\begin{equation}\label{eq:det_W}
\mathrm{det}(\hat{\boldsymbol{W}}_p) = \sum_{g=1}^G\tilde{c}_g(\boldsymbol{A},\boldsymbol{H}) \prod_{f=1}^F z_f^{q_{g,f}},
\end{equation}
for some $G$, with $\sum_{f=1}^{F} q_{g,f} = P, \; \forall g$. Thus, the only possibility for having $\mathrm{det}(\hat{\boldsymbol{W}}_p)=0$ is if the coefficients are zero, i.e., $\tilde{c}_g(\boldsymbol{A},\boldsymbol{H})=0, \, 1\leq g \leq G$. However, the coefficients $\tilde{c}_g(\boldsymbol{A},\boldsymbol{H})$ are rational expressions of the entries in $\boldsymbol{A}$ and $\boldsymbol{H}$. This means that, in order to have $\tilde{c}_g(\boldsymbol{A},\boldsymbol{H})=0$, a polynomial multi-variate expression of the entries in $\boldsymbol{A}$ and $\boldsymbol{H}$ must be 0. Again, this can only happen at most in a countable set of $\boldsymbol{A}$ and $\boldsymbol{H}$ as long as we find an $\boldsymbol{A}$ and an $\boldsymbol{H}$ such that $\mathrm{det}(\hat{\boldsymbol{W}}_p),\, \forall p$, while still assuring that $\boldsymbol{B}$ is full-rank so that $F$ remains fixed. The same example as in the proof of Lemma~\ref{lem:det_B}, i.e., randomly chosen $\boldsymbol{A}$ and $\{\boldsymbol{H} \}_{i,j}=\{\boldsymbol{A}\}_{i,j}$, $i\in \{ 1,\dots M \}$, $j \in \{ 1,\dots, K\}$, gives the trivial solution $\hat{\boldsymbol{W}}=\mathbf{I}_M$, and thus fulfills both conditions. Therefore, we have proved that a randomly chosen $\boldsymbol{A}$ and $\boldsymbol{H}$ will be in the set $\mathcal{V}$ with probability 1.



\section*{Appendix B: Proof of Lemma~\ref{rankrow}}
From the structure of $\boldsymbol{B}_1$ in \eqref{eq:B1}, we observe that a particular row of $\boldsymbol{A}$ appears exactly in $K$ rows of $\boldsymbol{B}$. Let us denote $\boldsymbol{B}_0$ as the submatrix of $\boldsymbol{B}$ formed by all rows in $\boldsymbol{B}$ where the rows of $\boldsymbol{A}_0$ appear. Clearly, to satisfy \eqref{eq:lin_wax}, we must in particular satisfy $\boldsymbol{B}_0 \boldsymbol{u}=\boldsymbol{0}_{RK\times 1}$. Now, $\boldsymbol{B}_0$ reads
$$\boldsymbol{B}_0=\left[\boldsymbol{I}_K \otimes \boldsymbol{A}_0 \;\;  \widehat{\boldsymbol{H}}_0\right],$$
where $\widehat{\boldsymbol{H}}_0$ is formed from $\boldsymbol{H}$ as follows: Let $\iota(r)$ denote the block $\boldsymbol{H}_{\iota(r)}$ where the $r$th row in $\boldsymbol{A}_0$ is taken from. Let $\boldsymbol{H}_0=\left[ \boldsymbol{H}_{\iota(1)}^T \; \boldsymbol{H}_{\iota(2)^T} \, \dots \, \boldsymbol{H}_{\iota(R)}^T \right]^T$, and let $\mathbb{I}\mathbb{I}(\ell)$ be an $R\times L$ matrix  with a single entry equal to 1 at row $\ell$ and column $(\iota(\ell) \;\mathrm{mod} \; L)+1$, and all other equal to 0. Then,
\begin{eqnarray} \label{hatH} \widehat{\boldsymbol{H}}_0&=&\left[\boldsymbol{0}_{\mathcal{D}_0} \;\; \boldsymbol{H}_{\iota(1)}^H \!\otimes\! \mathbb{I}\mathbb{I}(1) \; \; \boldsymbol{0}_{\mathcal{D}_1}  \;\; \;\boldsymbol{H}_{\iota(2)}^H \!\otimes\! \mathbb{I}\mathbb{I}(2) \; \right. \nonumber \\
&& \;\quad \quad \quad \left. \dots \quad \boldsymbol{0}_{\mathcal{D}_{R-1}}  \;\; \boldsymbol{H}_{\iota(R)}^H\! \otimes \!\mathbb{I}\mathbb{I}(R)  \;\; \boldsymbol{0}_{\mathcal{D}_R} \right]\end{eqnarray}
where we have used the shorthand notation 
\begin{eqnarray} 
&&\mathcal{D}_k=RK\times (\iota(k+1)-\iota(k))L^2, \nonumber \\
&&\iota(0)\triangleq 1,  \;\;\;\; \iota(R+1)\triangleq M/L. \nonumber 
\end{eqnarray}

To study the null space of $\boldsymbol{B}_0$ we may just as well study the null space of $(\boldsymbol{I}_K \otimes \boldsymbol{Q}_0^H)\boldsymbol{B}_0$, where $\boldsymbol{Q}_0\boldsymbol{R}_0 = \boldsymbol{A}_0$ is the QR decomposition of $\boldsymbol{A}_0$. We have, 
\begin{equation} \label{appb:qr}
(\boldsymbol{I}_K \otimes \boldsymbol{Q}_0^H)\boldsymbol{B}_0= \left[\boldsymbol{I}_K \otimes {\boldsymbol{R}}_0 \; \; (\boldsymbol{I}_K \otimes \boldsymbol{Q}_0^H)\widehat{\boldsymbol{H}}_0\right].\end{equation} Let $\kappa= \mathrm{rank}(\boldsymbol{A}_0)$. The matrix $\boldsymbol{I}_K \otimes {\boldsymbol{R}}_0$ consequently has $K(R-\kappa)$ all-zero rows. 

If we extract said all-zero rows, we obtain,$$
\left[\boldsymbol{0}_{K\!(\!R-\kappa)\times TK} \; \boldsymbol{P}(\boldsymbol{I}_K \otimes \boldsymbol{Q}_0^H)\widehat{\boldsymbol{H}}_0\right] \!\!\begin{bmatrix}
\mathrm{vec}({\boldsymbol{X}})\\
\mathrm{vec}(\boldsymbol{W}_1)\\
\vdots \\
\mathrm{vec}(\boldsymbol{W}_P)
\end{bmatrix}\!=\! \boldsymbol{0}_{K\!(\!R-\kappa)\times 1} 
$$
where $\boldsymbol{P}$ is a $K(R-\kappa)\times KR$ matrix that extracts the rows where $\boldsymbol{I}_K \otimes {\boldsymbol{R}}_0$ is all-zero. This implies that we can discard $\boldsymbol{X}$ so that we equivalently obtain
\begin{equation} \label{imop2}  \boldsymbol{P}(\boldsymbol{I}_K \otimes \boldsymbol{Q}_0^H)\widehat{\boldsymbol{H}}_0\begin{bmatrix}
\mathrm{vec}(\boldsymbol{W}_1)\\
\vdots \\
\mathrm{vec}(\boldsymbol{W}_P)
\end{bmatrix} = \boldsymbol{0}_{K(R-\kappa)\times 1}.\end{equation}

We next note that, due to the many all-zero columns in $\widehat{\boldsymbol{H}}_0$ (represented by the terms $\boldsymbol{0}_{\mathcal{D}_k}$ in \eqref{hatH}), not all the $\boldsymbol{W}_p$ matrices matter. In fact, it can be straightforwardly verified that \eqref{imop2} is equivalent to
\begin{equation} \label{imop3} \boldsymbol{P}(\boldsymbol{I}_K \otimes \boldsymbol{Q}_0^H)\bar{\boldsymbol{H}}_0 \begin{bmatrix}
\boldsymbol{w}_{\iota(1)} \\
\vdots \\
\boldsymbol{w}_{\iota(R)} 
\end{bmatrix} = \boldsymbol{0}_{K(R-\kappa)\times ML},\end{equation}
where $\boldsymbol{w}_{m}$ is the $1\times L$ vector formed from extracting the entries at the $m$th row of $\boldsymbol{W}$ that are allowed to take non-zero values, and
\begin{equation} \label{barH} \bar{\boldsymbol{H}}_0=\left[ \boldsymbol{H}_{\iota(1)}^H \!\otimes\! \mathbb{I}(1) \;  \;\boldsymbol{H}_{\iota(2)}^H \!\otimes\! \mathbb{I}(2)  \;\dots \; \boldsymbol{H}_{\iota(R)}^H\! \otimes \!\mathbb{I}(R) \right],\end{equation}
where $\mathbb{I}(\ell)$ is the non-zero column of $\mathbb{I}\mathbb{I}(\ell)$.

For randomly chosen $\boldsymbol{H}$, the matrix $\boldsymbol{P}(\boldsymbol{I}_K \otimes \boldsymbol{Q}_0^H)\bar{\boldsymbol{H}}_0$ is full rank with probability 1. Therefore, \eqref{barH} only has a non-trivial solution whenever the number of unknowns is larger than the number of equations, i.e., whenever, $RL > K(R-\kappa)$. Consequently, 
$$\kappa > R\frac{K-L}{K}.$$


\section*{Appendix C: Necessary condition for $\boldsymbol{A}$}
As per previous results, a matrix $\H$ can, with probability 1, be decomposed as  $\H=\WAX$ if and only if there exists a full rank matrix $\hat{\W}$ such that $ \AX=\hat{\W}\H$. We are now interested in establishing necessary conditions for the matrix $\A$ so that this is possible. Let $\As$ denote a submatrix of $\A$ comprising an arbitrary selection of rows in $\A$ with $\ra{\As}=r.$ An immediate consequence is that $\ra{\As \X}\leq r.$
Assuming that $\H=\WAX$ holds,  $\ra{\hat{\Ws}\H}=\ra{\As \X}$,  where $\hat{\Ws}$ is a submatrix of $\hat{\W}$ comprising rows corresponding to those in $\As$. On the other hand, if no $\hat{\Ws}$ exists such that $\ra{\hat{\Ws}\H} \leq r$, we can infer that $\A$ does not allow for a WAX decomposition, i.e., there are no matrices $\W$ and $\X$ such that $\H=\WAX$. Therefore, a necessary condition on $\A$ for the existence of a WAX decomposition is
 \be \ra{\As} \geq \min_{\hat{\Ws} \in \mathcal{W}} \ra{\hat{\Ws}\H},\quad \forall \mathcal{S},  \label{eq:rankproof1} \ee
where we will define the set $\mathcal{W}$ after having introduced further notation. The matrix $\A$ contains $M$ rows, and since $L$ divides $M$, $\A$ contains $P=M/L$ blocks of $L$ rows.  Let the $1\times P$ vector $\vec{a}$ denote the number of rows in $\As$ taken from the $p$th block in $\A$, and let ${\As}_{\!,p}$  and ${\hat{\Ws}}_{\!,p}$ be the $a_p \times T$ and $a_p \times M$ submatrices  of $\As$ and $\hat{\Ws}$, respectively, formed  from these $a_p$ rows. With that, $\mathcal{W}$ is the set containing all block-diagonal matrices where the $p$th block is of dimensions $a_p\times K$ and has rank $a_p$ (the latter is needed to ensure that the overall matrix $\hat{W}$ is invertible).

We are now ready to study \eqref{eq:rankproof1}. Suppose that the minimum of \eqref{eq:rankproof1} is $r$. Since $\hat{\Ws}\H$ has dimensions $(\sum_p a_p) \times K$ this implies that the null-space of $\hat{\Ws}\H$ has dimension $K-r$. Thus, if  $\ra{\hat{\Ws}\H}=r$, there must exist a $K\times (K-r)$ matrix $\vec{N}$ such that $\hat{\Ws}\H\vec{N}=\vec{0}$. Recalling that $\hat{\Ws}$ is block-diagonal with each block being of dimension $a_p \times L$ and having rank $a_p$, it follows that 
\be \ra{\H _p \vec{N}} \leq L-a_p \label{eq:rankproof2} \ee
where $\H _p$ is of dimension $L\times K$ and  $\H=[\vec{H}_1^\mathrm{H} \; \cdots \; \vec{H}_P^\mathrm{H}]^\mathrm{H}$. 
Thus, if we want to solve the minimization in \eqref{eq:rankproof1}, we have the equivalent problem of finding the maximum possible number of columns in the matrix $\vec{N}$ so that it is full rank and  $ \ra{\H _p \vec{N}} \leq L-a_p, \; p=1\ldots P$.

We next define the rank-profile $\vec{J}$ of a matrix $\vec{N}$ as  ${J_{b,p}=\ra{\H _p \vec{N}_{1:b}}-\ra{\H _p \vec{N}_{1:b-1}}}$, where $\vec{N}_{1:b}$ denotes columns $1, 2, \ldots , b$ of $\vec{N}$. We remark that $J_{b,p}\in\{0,1\}$, and that $\vec{J}$ depends on $\vec{H}$, although our notation does not indicate this dependency. In the following lemma, we characterize the admissible rank-profiles.
\begin{lem}\label{lem:rank_prof}
There exists a full rank matrix $\vec{N}$ with $B$ columns such that $ \ra{\H _p \vec{N}} \leq L-a_p$ if and only if the rank-profile matrix $\vec{J}$ satisfies
$$\begin{aligned}
K\!-\!Ld_{<L}(\bar{\vec{a}})+\!\!\sum_{b^\prime =1}^{b-1}\sum_{p=1}^P J_{b^\prime,p} + \!\! \sum_{p=1}^P \!  J_{b,p} \! \left(L-\sum_{b^\prime =1}^{b-1}   J_{b^\prime,p}    \right) \! \geq       b ,&\\
 1\leq b \leq B,&
\end{aligned}$$
$$ \sum_{b=1}^B J_{b,p}\leq \bar{a}_p,\quad 1\leq p \leq P,$$
where $d_{<L}(\cdot)$ denotes the number  of elements of its argument that are less than $L$, and $\bar{a}_p=L-a_p$.
\end{lem}
\begin{proof}
We first note that if $a_p=0$ for any $p$, then the condition in \eqref{eq:rankproof2} is trivially satisfied for any $\vec{N}$ so we can without loss of generality assume that $a_p>0,\; \forall p$. From the definition of $J_{b,p}$ if follows that $\ra{\H _p \vec{N}}=\sum_{b=1}^B J_{b,p}$; thus, the second set of conditions in the lemma is trivial.

Consider now a specific value $b$. If $J_{b,p}=0$, it follows that column $b$ of $\vec{N}$, $\vec{n}_b$ is restricted to
$$\vec{n}_b\in \mathcal{I}_{b,p}=\mathcal{N}(\H _p)\cup \mathrm{span}\left(\H _p ^+ \H _p \vec{N}_{1:b-1}\right),$$
while $J_{b,p}=1$ does not restrict $\vec{n}_b$ so that $\vec{n}_b\in \mathcal{I}_{b,p}=\mathbb{C}^K.$
By inspection, it can be seen that the dimensionality of $\mathcal{I}_{b,p}$ can be written as
$$\mathrm{dim}(\mathcal{I}_{b,p})=K(1-J_{b,p})+J_{b,p}\left(K-L+\sum_{b^\prime=1}^{b-1} J_{b,p}\right).$$
Altogether, we have that $$\vec{n}_b \in \bigcap_{p=1}^{P} \mathcal{I}_{b,p}.$$
The dimensionality of the intersection satisfies
$$\mathrm{dim}\left( \bigcap_{p=1}^{P} \mathcal{I}_{b,p}\right) = \sum_{p=1}^P \mathrm{dim}(\mathcal{I}_{b,p}) -(P-1)K.$$ 
This dimensionality must be at least b, since there are already $b-1$ vectors in $\vec{N}_{1:b-1}$ in the same space, and the rank of $\vec{N}_{1:b}$ must be full for all $b$. By manipulation, the statement of the lemma can be obtained.
\end{proof}

Altogether, a necessary condition on $\A$ is that the rank of any submatrix $\As$ satisfies
$\ra{\As}\geq K-B_{\mathcal{S}}$, where $B_{\mathcal{S}}$ is the largest integer $B$ that satisfies the conditions specified in Lemma~\ref{lem:rank_prof}; note that the vector $\vec{a}$ in the conditions depends on $\mathcal{S}$. It can be seen from Lemma~\ref{lem:rank_prof} that finding the largest integer $B$ that satisfies the conditions is a non-linear (quadratic) integer problem and we have not been able to solve it in closed form. 
\bibliographystyle{IEEEtran}
\bibliography{IEEEabrv,wax}

%
\IEEEpeerreviewmaketitle

\end{document}